\documentclass[12pt, draftclsnofoot, onecolumn]{IEEEtran}
\usepackage{amsfonts}
\usepackage{cite}
\usepackage[top=1 in, bottom=1 in, left=1 in, right= 1 in]{geometry}
\usepackage[cmex10]{amsmath}
\usepackage{amssymb}
\usepackage{array}
\usepackage{chemarrow}
\usepackage{times,epsfig}
\usepackage{graphicx}
\usepackage{subcaption}
\usepackage{setspace}
\usepackage{times,epsfig}
\usepackage{bbm}
\usepackage{verbatim}
\usepackage{amsmath}
\usepackage{epstopdf}
\usepackage{mathtools}

\newtheorem{remark}{\textbf{Remark}}

\newtheorem{theorem}{\textbf{Theorem}}
\newtheorem{lemma}{\textbf{Lemma}}

\makeatletter
\g@addto@macro{\normalsize}{%
	\setlength{\abovedisplayskip}{3pt plus 2pt minus 2pt}
	\setlength{\abovedisplayshortskip}{3pt plus 2pt minus 2pt}
	\setlength{\belowdisplayskip}{3pt plus 2pt minus 2pt}
	\setlength{\belowdisplayshortskip}{3pt plus 2pt minus 2pt}
	\setlength{\textfloatsep}{10pt plus 2pt minus 2pt}
}
\makeatother

\begin{document}

\title{Design and Analysis of Initial Access in Millimeter Wave Cellular Networks}

\author{ Yingzhe Li, Jeffrey G. Andrews, Fran\c{c}ois Baccelli, Thomas D. Novlan, Jianzhong Charlie Zhang \thanks{Y. Li, J. G. Andrews and F. Baccelli are with the Wireless Networking and Communications Group (WNCG), The University of Texas at Austin (email: yzli@utexas.edu, jandrews@ece.utexas.edu, francois.baccelli@austin.utexas.edu). T. Novlan is with AT\&T Labs (email: tdnovlan@utexas.edu). J. Zhang is with Samsung Research America-Dallas (email: jianzhong.z@samsung.com). Part of this paper was presented at IEEE Globecom 2016, workshop on emerging technologies for 5G wireless cellular networks~\cite{li2016mmWIAGC}. Date revised: \today. }}

\maketitle

\begin{abstract}
	Initial access is the process which allows a mobile user to first connect to a cellular network.  It consists of two main steps: cell search (CS) on the downlink and random access (RA) on the uplink. Millimeter wave (mmWave) cellular systems typically must rely on directional beamforming (BF) in order to create a viable connection. The beamforming direction must therefore be learned -- as well as used -- in the initial access process for mmWave cellular networks.  This paper considers four simple but representative initial access protocols that use various combinations of directional beamforming and omnidirectional transmission and reception at the mobile and the BS, during the CS and RA phases. We provide a system-level analysis of the success probability for CS and RA for each one, as well as of the initial access delay and user-perceived downlink throughput (UPT).  For a baseline exhaustive search protocol, we find the optimal BS beamwidth and observe that in terms of initial access delay it is decreasing as blockage becomes more severe, but is relatively constant (about $\pi/12$) for UPT. Of the considered protocols, the best trade-off between initial access delay and UPT is achieved under a fast cell search protocol. 
\end{abstract}

\section{Introduction}
Initial access refers to the procedures that establish an initial connection between a mobile user and the cellular network, and is a critical prerequisite for any subsequent communication.  The design of initial access is a central challenge for mmWave cellular systems relative to existing cellular systems, for two main reasons.  First, mmWave links generally require high directionality (i.e. large antenna gain) to achieve a sufficient signal-to-noise ratio (SNR)~\cite{andrews2014will,pi2011introduction,rappaport2013millimeter,ghosh2014millimeter}. But the mobile and the BS have no idea what directions to use upon initial access. Thus, they must search over a large beamforming space to find each other, which is time consuming. Second, because mmWave links are vulnerable to blocking and falling out of beam alignment, initial access will need to be done much more frequently than in conventional systems. 
Because of the importance of initial access for mmWave cellular systems, significant effort is currently underway to design efficient protocols. In this paper, we develop a general analytical framework and detailed performance analysis for initial access in a mmWave cellular system, by leveraging stochastic geometry \cite{stochtutorial,baccelli2010stochastic,chiu2013stochastic,trac, mukherjee2012distribution}. We believe the analytical tools developed in this paper can be extended to a wide variety of initial access protocols, to provide a useful complement to simulation-based evaluations and to help optimize key parameters such as the number of beams.

\subsection{Prior Work}
Initial access, for mmWave specifically, has been investigated by a few standard organizations in recent years~\cite{ieee2012IEEEad,Nitsche2014IEEE,Verizon20165G2}. The IEEE 802.11ad standard adopted a two level initial beamforming training protocol for the 60 GHz band, where a coarse-grained sector level sweep phase is followed by an optional beam refinement phase~\cite{ieee2012IEEEad,Nitsche2014IEEE}. However, IEEE 802.11ad is mainly designed for indoor communications within an ad hoc type network. 
The Verizon 5G forum~\cite{Verizon20165G2} has created technical specifications for early mmWave cellular systems, where beam sweeping is applied by the BSs during cell search, and the beam reference signal (BRS) is transmitted to enable the users to determine appropriate BS beamforming directions.

Despite the standardization of initial access for mmWave cellular networks is still in its early stages, several recent research efforts have investigated this problem~\cite{Jeong2016random,Desai2014initial,Giordani2016comparative,Raghavan2016beamforming,barati2015directional,Shokri2015Millimeter}. 
An exhaustive procedure to sequentially search all the possible transmit-receive beam pairs has been proposed in~\cite{Jeong2016random}. A hierarchical search procedure is proposed in~\cite{Desai2014initial}, where the BS first performs an exhaustive search over wide beams, then refines to search narrow beams. The exhaustive and hierarchical strategies are compared in~\cite{Giordani2016comparative}, which shows that hierarchical search generally has smaller initial access delay, but exhaustive search gives better coverage to cell-edge users. 
By adapting limited feedback-type directional codebooks, a low-complexity beamforming approach for initial user discovery is proposed in~\cite{Raghavan2016beamforming}. Several initial access options with different modifications to LTE initial access procedures 
are proposed in~\cite{barati2015directional}, which has observed that the initial access delay can be reduced by omni-directional transmission from the BSs during cell search. 
A cell search procedure that leverages synchronization from the macro BSs, followed by sequential spatial search from the mmWave BSs, is shown to enhance the initial access efficiency~\cite{Shokri2015Millimeter}. 

All the aforementioned works are either a point-to-point analysis or only consider one user with a few nearby BSs and, 
a system-level analysis of initial access in mmWave networks
has yet to be offered. In recent years, stochastic geometry has been recognized as a powerful mathematical tool to analyze performance of large-scale mmWave cellular networks~\cite{akoum2012covrage,singh2015tractable,Jihong2016Tractable,bai2015coverage,alkhateeb2016initial,DiRenzo2015Stochastic}. By incorporating directional beamforming without capturing the blockage effects, \cite{akoum2012covrage} shows mmWave network can achieve comparable coverage and much higher data rate than conventional microwave networks. Similar performance gains of mmWave networks have been observed when statistical blockage models are used, such as a line-of-sight (LOS) ball blockage model~\cite{bai2015coverage,singh2015tractable,Jihong2016Tractable}, an exponential decreasing LOS probability function with respect to (w.r.t.) the link length~\cite{bai2015coverage,alkhateeb2016initial}, or a blockage model which also incorporates an outage state~\cite{DiRenzo2015Stochastic}. 
However,~\cite{akoum2012covrage,bai2015coverage,singh2015tractable,Jihong2016Tractable,alkhateeb2016initial,DiRenzo2015Stochastic} all assume the association between user and its serving BS has already been established, while in fact the initial access is a key challenge and performance limiting factor for mmWave networks. 

\subsection{Contributions}
In this work, four initial access protocols are investigated, including a baseline exhaustive search protocol wherein BS and user sweep through all transmit-receive beam pairs during cell search, and three protocols that require less overhead. Our main contributions are as follows:

\textbf{Accurate analytical framework for mmWave system-level performance under various initial access protocols.} Different from the link-level analysis in~\cite{Jeong2016random,Desai2014initial,Giordani2016comparative,Raghavan2016beamforming,barati2015directional}, we derive several system-level performance metrics in mmWave cellular network for the first time, including the expected initial access delay, and a new metric called average user-perceived downlink throughput which quantifies the effect of the initial access protocol on the user-perceived throughput performance. Our analytical results are validated against detailed system level simulations.

\textbf{Beam sweeping is shown to be essential for cell search.} We find that the mmWave system is subject to significant coverage issues if beam sweeping is not applied during cell search. By contrast, a reasonable cell search success probability can be achieved even with a small (e.g., 4 to 8) number of beamforming directions to search at the BS or user. 

\textbf{A detailed performance evaluation for the baseline exhaustive search protocol.} The baseline protocol is shown to have low random access preamble collision probability irrespective of the blockage conditions. An optimal BS beamwidth in terms of the expected initial access delay is found, which decreases as blockage becomes more severe. In addition, 
the optimal BS beamwidth in terms of the average user-perceived downlink throughput does not vary too much for different blockage conditions, and is typically within $[10^\circ, 18^\circ]$ in our evaluations.  

\textbf{Comparison of expected initial access delay and average user-perceived downlink throughput.} The baseline exhaustive search protocol gives the best initial access delay performance when blockage is severe, but it also has the worst user-perceived downlink throughput, due to its high initial access overhead. By contrast, the protocol wherein the BS (user) applies beam sweeping and the user (BS) receives omni-directionally during cell search (random access), generally gives the best user-perceived downlink throughput performance but has high initial access delay. Of the four considered sample protocols, the best trade-off between initial access delay and average user-perceived downlink throughput is achieved when the BS transmits using wide beams and the user applies beam sweeping during cell search. 



\section{System Model}
In this work, a time-division duplex (TDD) mmWave system in Fig.~\ref{fig:time_Structure} is considered, where the system time is divided into different initial access cycles with period $T$. Each cycle begins with a cell search phase, followed by the random access phase and the data transmission phase. The mmWave cellular system has carrier frequency $f_c$ and total system bandwidth $W$. The transmit power of BSs and users are $P_b$ and $P_u$ respectively, and the total thermal noise power is $\sigma^2$. In the rest of this section, we present the spatial location models, the propagation and blockage assumptions, and the antenna and beamforming models. 

\begin{figure}
	\centering
	\includegraphics[width=0.6\linewidth]{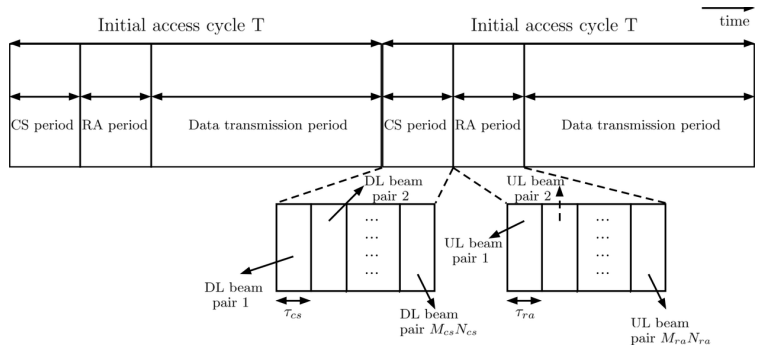}
	\caption{Illustration of two cycles for the timing structure.}
	\label{fig:time_Structure}
\end{figure}

\subsection{Spatial Locations}\label{SpatLocSubsec}
The locations for BSs and users are modeled as realizations of two independent homogeneous Poisson point processes (PPPs). Specifically, the BS process $\Phi = \{x_i\}_i$ has intensity $\lambda$, and the user process $\Phi_u = \{u_i\}_i$ has intensity $\lambda_u$. The PPP assumption for BS locations could lead to many tractable and insightful results. In fact, this assumption is also reasonable since~\cite{guo2014ADG} has proved that the SINR trend under the PPP assumption only has a constant SINR gap compared to any other stationary BS location model. 
In addition,~\cite{blaszczyszyn2015wireless} has proved that for any arbitrary spatial BS location pattern with sufficiently large shadowing variance, the statistics of the propagation losses of a user with respect to all BSs will converge to that of a Poisson network. As a result, the PPP assumption for BSs can also be treated by combining the shadowing effects and the BS locations. Thus we do not consider shadowing separately in our analysis, similar to~\cite{andrews2017modeling,bai2015coverage,alkhateeb2016initial,gupta2015feasibility,bai2014coveragemag,Jihong2016Tractable}.

Since the user locations form a realization of a PPP, we can analyze the performance of a typical user located at the origin. This is guaranteed by Slivnyak's theorem, which states that the property observed by the typical point of PPP $\Phi^{'}$ is the same as that observed by the point at origin in $\Phi^{'} \cup \{o\}$~\cite{chiu2013stochastic,baccelli2010stochastic}. 

\begin{table}[t]
\center\caption{Definitions and Values for System Parameters}\label{SysParaTable}
\resizebox{450pt}{!}{%
	\begin{tabular}{|c|p{155mm}|p{58mm}|}
		\hline 
		Symbol & Definition & Simulation Value \\ 
		\hline 
		$\Phi$, $\lambda$ & MmWave BS PPP and intensity & 100 BS/km$^2$\\ 
		\hline 
		$\Phi_u$, $\lambda_u$ & User PPP and intensity & 1000 user/km$^2$ \\ 
		\hline
		$\Phi_L$, $\Phi_N$ & LOS and NLOS BS tier to the typical user & \\ 
		\hline
		$f_c, W$ & Carrier frequency and system bandwidth & 28 GHz, 100 MHz \\
		\hline 
		$P_b$, $P_u$ & BS and user transmit power & 30 dBm, 23 dBm \\ 
		\hline
		$\sigma^2$ & Total thermal noise power& -94 dBm\\
		\hline
		$G(\theta),g(\theta)$ & Main lobe and side lobe gain at BS and user with beamwidth $\theta$, defined in~(\ref{BFGainModel})& $C_0 = 10$ dB for user antennas \\ 
		\hline
		$M, N$ & Number of antennas/BF directions at each BS and user & $M = 4, 8, ..., 48$, $N = 4$\\
		\hline 
		$M_{cs}$, $N_{cs}$, $K_{cs}$ & $M_{cs}$/$N_{cs}$: number of BF directions to search at BS/user in cell search; $K_{cs} = \min (M_{cs}, N_{cs})$ & \\ 
		\hline
		$m_{cs}$ & Number of wide beams to sweep during cell search for fast CS protocol & $4$\\
		\hline
		$M_{ra}$, $N_{ra}$  & Number of BF directions to search at BS and user during RA & \\ 
		\hline 
		$N_{pa}$ & Number of random access preamble sequences & 64 \\ 
		\hline 
		$\alpha_L, \alpha_N$ & Path loss exponent for LOS and NLOS links & 2, 4 \\ 
		\hline 
		$\beta$ & Path loss at close-in reference distance (i.e., 1m)  & 61.4 dB \\ 
		\hline 
		$\Gamma_{cs}, \Gamma_{ra}$ & SINR threshold to detect synchronization signal and RA preamble & -4 dB,  -4 dB  \\ 
		\hline
		$\tau_{cs}, \tau_{ra}$ & Duration for synchronization signal and RA preamble sequence & 14.3 $\mu$s, 14.3 $\mu$s\\
		\hline
		$T$ & Initial access cycle period & 20 ms \\
		\hline 
		$h(r)$ & Probability for a link with length $r$ to be LOS  & \\
		\hline
		$R_c$, $p$ & Radius and LOS probability for the LOS region in the LOS ball model & $R_c = 100$m, $p = 1, 0.75, 0.5, 0.25$ \\
		\hline
		$\mu$ & LOS region size for the exponential blockage model & $\mu = 100$m, $50$m, $25$m \\
		\hline
		$B(x,r)$ ($B^{o}(x,r)$) & Closed (open) ball with center $x$ and radius $r$ & \\
		\hline
		$S(u,\theta_1,\theta_2)$ & Infinite sector domain $\{ x \in \mathbb{R}^2, \text{ s.t., } \angle(x - u) \in [\theta_1, \theta_2)\}$ & \\
		\hline
		$S_j$ ($1 \leq j \leq K_{cs}$)& The $j$-th BS locatoin location sector with $S_j \triangleq S(o,\frac{2\pi (j-1)}{K_{cs}}, \frac{2\pi j}{K_{cs}})$ & \\
		\hline
		$V(z,T,\lambda) , U(z,T,\lambda) $ & Two special functions defined in~(\ref{SpecFnDefnEq}) & \\
		\hline
		$f_{Z_1}(z)$ & The PDF for the minimum path loss from the typical user to BSs inside the typical BS sector, which is given by~(\ref{MinPLPerSectorEq}) & \\
		\hline
		$P_{M_{cs},N_{cs}} (\Gamma_{cs})$ & Probability to detect the BS providing the smallest path loss inside the typical BS sector derived, which is given by~(\ref{CSProbPerSecEq}) &\\
		\hline
		$\tilde{P}_{M_{cs},N_{cs}}(z,\Gamma_{cs})$ & $\tilde{P}_{M_{cs},N_{cs}}(z,\Gamma_{cs})$: Conditional detection probability when the minimum path loss inside the typical BS sector is $z$, which is given by~(\ref{GFuncDefnEq}) &\\
		\hline
		$P_{M_{cs},N_{cs}}(z_0,\Gamma_{cs})$ & $P_{M_{cs},N_{cs}}(z_0,\Gamma_{cs}) = \int_{z_0}^{\infty} \tilde{P}_{M_{cs},N_{cs}}(z,\Gamma_{cs}) {\rm d}z$ &\\
		\hline
		$P_{co}$ & Probability of no RA preamble collision, given by~(\ref{Prob_No_PA_Collision_Eq}) &\\
		\hline
		$P_{ra} (Z_0,\Gamma_{ra})$ & Probability the RA preamble SINR at the tagged BS exceeds $\Gamma_{ra}$, given by~(\ref{RA_SINR_COP_Lemma_Eq}) & \\
		\hline
		$\eta_{IA}$ & Overall success probability of initial access, given by~(\ref{IA_SP_Eq}) & \\
		\hline
	\end{tabular} }
\end{table}

\subsection{Blockage and Propagation Models}
The link between a BS and a user is either line-of-sight (LOS) or non-line-of-sight (NLOS). We denote by $h(r)$ the probability for a link of distance $r$ to be LOS, which is only a function of $r$ and independent of other links.  
From the typical user's perspective, the BS process $\Phi$ is divided into two tiers: the LOS BS tier $\Phi_L$ and the NLOS BS tier $\Phi_N$. Since the LOS probability function $h$ only depends on the link length, $\Phi_L$ and $\Phi_N$ are two independent PPPs. For any $x \in \mathbb{R}^2$, the intensity function for $\Phi_L$ is $\lambda_L(x) = \lambda h(\|x\|)$, and the intensity function for $\Phi_N$ is $\lambda_N(x) = \lambda (1-h(\|x\|)).$ Incorporating the blockage model to differentiate the LOS and NLOS links is the most distinctive difference for analyzing the mmWave network performance, compared to the analysis in traditional sub-6 GHz networks~\cite{trac}.

Two examples of LOS probability functions $h(r)$ include: (1) the ``generalized LOS ball model"~\cite{bai2015coverage,singh2015tractable} with $h(r) = p \mathbbm{1}_{r \leq R_c}$, where $R_c$ represents the radius for the LOS region ($R_c > 1$m), and $p$ represents the LOS probability within the LOS region; (2) the ``exponential blockage model"~\cite{bai2015coverage} with $h(r) = \exp(-r / \mu)$, where $\mu$ represents the average LOS region length. Compared to the 3GPP blockage model which has accurate fit to the empirical LOS probability,~\cite{kulkarni2014coverage,andrews2017modeling} show that the LOS ball model and exponential blockage model better estimate the SINR and are simpler. 

The path loss for a link with distance $r$ in dB is given by: 
\allowdisplaybreaks
\begin{align}
\allowdisplaybreaks
&\mathit{l}(r)=
\left\{
\begin{array}{ll}
10\log(\beta) + 10 \alpha_L \log_{10}(r) \text{    dB}, & \text{if LOS},\\
10\log(\beta) + 10 \alpha_N \log_{10}(r) \text{    dB}, & \text{if NLOS},
\end{array}
\emph{ } \right.
\end{align}
where $\alpha_L$ and $\alpha_N$ represent the path loss exponent for LOS and NLOS links respectively, and $\beta$ is the path loss at a close-in reference distance (i.e., 1 meter). For the rest of the paper, the path loss function for LOS link and NLOS link are denoted by $\mathit{l}_L(r)$ and $\mathit{l}_N(r)$ respectively.

The small scale fading effect is assumed to be Rayleigh fading, where each link is subject to an i.i.d. exponentially distributed fading power with unit mean. Compared to more realistic small-scale fading models such as Nakagami-$m$ fading, Rayleigh fading leads to much more tractable results with very similar design insights~\cite{gupta2015feasibility, Jihong2016Tractable,andrews2017modeling}. 

\subsection{Antenna Model and Beamforming Gains}\label{AntModSubSec}
BSs and users are equipped with an antenna array of $M$ and $N$ antennas respectively to support directional communications, where $M/N \in \mathbb{N}^+$. Both mmWave BSs and users have 1 RF chain, such that only one analog beam can be transmitted or received at a time\footnote{Our analysis in the rest of the paper based on analog beamforming directly applies to the scenario where cell sectorization with frequency reuse across sectors within the same cell is used. Hybrid beamforming is left to future work.}. For analytical tractability, we assume the actual antenna pattern is approximated by a sectorized beam pattern~\cite{bai2015coverage,akoum2012covrage,singh2015tractable,kulkarni2014coverage,alkhateeb2016initial,gupta2015feasibility,DiRenzo2015Stochastic,andrews2017modeling,bai2014coveragemag,Shokri2015Millimeter}, where the antenna has constant main-lobe gain over its half power beamwidth, and also a constant side-lobe gain otherwise. We adopt the beamforming gain model for sectorized beam pattern as~\cite{Shokri2015Millimeter,alkhateeb2016initial}, whose accuracy has been validated in Fig. 8 of~\cite{alkhateeb2016initial}. Specifically, if we denote by $G_u(\theta_u)$ the beamforming gain at user with beamwidth $\theta_u$, then $G_u(\theta_u)$ is given by:
\begin{align}
\allowdisplaybreaks
&G_u(\theta_u) =
\left\{
\begin{array}{ll}
G(\theta_u) = \frac{2\pi}{\theta_u} \frac{\gamma}{\gamma+1} , & \text{in the main lobe},\\
g(\theta_u) = \frac{2\pi}{2\pi-\theta_u} \frac{1}{\gamma+1} , & \text{in the side lobe},
\end{array}
\emph{ } \right.
\label{BFGainModel}
\end{align}
where $\gamma$ mimics the front-back power ratio, which is given by $\gamma = \frac{2\pi }{C_0(2\pi-\theta_u)}$ for some constant $C_0$. 
A similar beamforming gain model is used at BS, but we assume $0$ side lobe gain for BS (i.e. $\gamma$ in~(\ref{BFGainModel}) is extremely large), and thus the main lobe gain for BS with beamwidth $\theta_b$ is $G(\theta_b) = \frac{2\pi}{\theta_b}$. 
This assumption is important to ensure the analytical tractability in Section~\ref{SPCSRASec}, which is also reasonable since mmWave BSs use large dimensional antenna array, and modern antenna design could enable a front-to-back ratio larger than 30 dB for mmWave BSs~\cite{waterhouse2002broadband}. 


Similar to~\cite{barati2015directional,alkhateeb2016initial}, we assume each BS has a codebook of $M$ possible beamforming vectors, which will correspond to $M$ sectorized beam patterns that have non-overlapping main lobes with beamwidth $\frac{2\pi}{M}$. The $m$-th BS beam ($1 \leq m \leq M$) covers a sector area centered at the BS, whose angle is within $ [2\pi \frac{m-1 }{M}, 2\pi\frac{m}{M})$. The spatial signature of any plane wave of the BS is given by the superposition of these $M$ non-overlapping beam directions~\cite{barati2015directional}. Similarly, each user has $N$ possible sectorized-pattern beamforming vectors that correspond to $N$ non-overlapping main lobes with beamwidth $\frac{2\pi}{N}$. Fig.~\ref{fig:BF_Structure} shows the first beam direction of the user and the fifth beam direction of the BSs with $M = 8$ and $N = 4$. 
\begin{figure}[h]
	\centering
	\includegraphics[width=0.45\linewidth]{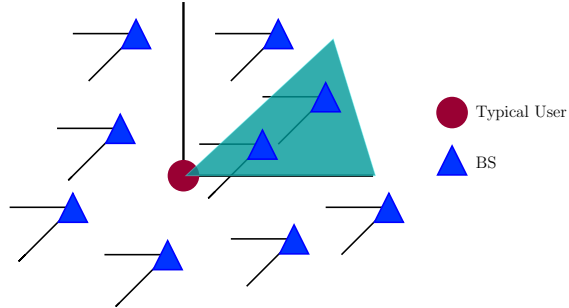}
	\caption{Beam pattern for BS and user beam pair $(5,1)$ with $M = 8, N = 4$. Only the typical user is shown, and the shaded area represents the corresponding BS sector.}
	\label{fig:BF_Structure}
\end{figure}
For any BS and its associated users, their aligned beamforming vectors need to be learned through the initial access, which will be used for subsequent data transmissions.

Finally, Table~\ref{SysParaTable} summarizes the definitions and simulation values of the important notation and system parameters that will be used in the rest of this paper. In particular, the two special functions $V(z,T,\lambda)$ and $U(z,T,\lambda)$ from $\mathbb{R}^2 \times \mathbb{R}^{+} \times \mathbb{R}^{+}$ to $\mathbb{R}^{+}$ are defined as:
\allowdisplaybreaks
\begin{align}\label{SpecFnDefnEq}
\allowdisplaybreaks
V(z,T,\lambda)  &=  \exp\biggl\{-2\pi\lambda \biggl(\int_{\mathit{l}_L^{-1}(z)}^{+\infty} \frac{Tzh(r)r{\rm d}r}{Tz + \mathit{l}_L(r)} + \int_{\mathit{l}_N^{-1}(z)}^{+\infty} \frac{Tz(1-h(r))r{\rm d}r}{Tz + \mathit{l}_N(r)} \biggl)\biggl\}, \nonumber\\
U(z,T,\lambda)  &= \exp\biggl\{-2\pi\lambda \biggl(\int_{0}^{+\infty} \frac{Tzh(r)r{\rm d}r}{Tz + \mathit{l}_L(r)} + \int_{0}^{+\infty} \frac{Tz(1-h(r))r{\rm d}r}{Tz + \mathit{l}_N(r)} \biggl)\biggl\}.
\end{align}

\section{Initial Access Design for MmWave Networks and Performance Metrics} \label{IASect}
Similar to~\cite{barati2015directional}, we investigate mmWave initial access protocols that are compliant with the basic procedures of LTE. However, the initial access of LTE is performed omni-directionally, which cannot be directly applied by mmWave networks due to the high isotropic path loss.


\subsection{Cell Search and Random Access Procedure}\label{CSSubSec}
We assume the BS and user will follow the beam patterns described in Section~\ref{AntModSubSec}. 
The two main design objectives for initial access in mmWave cellular networks include: (1) connect the users to the network, and (2) enable both BS and its associated user to learn their aligned beamforming directions with beamwidth $\frac{2\pi}{M}$ and $\frac{2\pi}{N}$ respectively. 
These objectives are achieved through the following directional cell search and random access procedures. 

During cell search phase, BSs sweep through $M_{cs}$ transmit beamforming directions to broadcast the synchronization signals, while users sweep through $N_{cs}$ receive beamforming directions to detect the synchronization signals. 
A synchronous beam sweeping pattern is used, such that during any synchronization signal period, all BSs/users will transmit/receive in the same direction, and one particular downlink beam pair is searched. 
We assume each user is able to detect a BS with negligible miss detection probability (e.g., less than 1\%)\footnote{A different miss detection probability can be assumed, but the performance trends derived in the paper will not be affected.}, if the signal-to-interference-plus-noise ratio (SINR) of the synchronization signal from that BS exceeds $\Gamma_{cs}$. The beam reference signal is assumed to be transmitted along with the synchronization signal, such that the user is able to acquire the BS beam direction upon successful cell search~\cite{Verizon20165G2}. Among all the BSs that are detected during CS, the user selects the BS that provides the smallest path loss as its serving BS. If cell search fails, the user will not transmit in the random access phase, and it needs to repeat the initial access procedure in the next cycle. Denote by $\tau_{cs}$ the duration for each synchronization signal, the cell search phase leads to a total delay of $T_{cs} = M_{cs} N_{cs} \tau_{cs}$, as shown in Fig.~\ref{fig:time_Structure}. 

In the random access phase, the user initiates the connection to its desired serving BS by transmitting a RA preamble sequence, which is uniformly selected from $N_{pa}$ orthogonal preamble sequences. Users sweep through $N_{ra}$ transmit beamforming directions synchronously, and BSs sweep through $M_{ra}$ receive beamforming directions synchronously during random access. The user can be discovered by its serving BS if: 1) there is no RA preamble collision with other users transmitting simultaneously to the same BS; and 2) the SINR of the preamble sequence exceeds $\Gamma_{ra}$. 
Similar to cell search, we assume the probability of miss detection is sufficiently small (e.g., less than 1\%) if the SINR of RA preamble exceeds $\Gamma_{ra}$.
The user is connected to its serving BS upon successful random access, and both the user and BS are aware of their beamforming directions for data transmission. 
According to Fig.~\ref{fig:time_Structure}, the total random access delay is $T_{ra} = M_{ra} N_{ra} \tau_{ra}$, with $\tau_{ra}$ representing the duration for each RA preamble sequence.

\subsection{Protocols for BS and User to Determine Beamforming Directions}
Different initial access protocols can be designed to enable the user and its serving BS to determine their aligned beamforming directions with beamwidth $\frac{2\pi}{N}$ and $\frac{2\pi}{M}$ respectively. The protocols that are investigated in this paper are as follows: 

1)  \emph{Baseline}: BSs and users sweep through all possible beamforming directions during cell search (i.e., $M_{cs} = M$, $N_{cs} = N$), so that the user can determine its beamforming direction after successful cell search. During random access, the user transmits in the beamforming direction it found during CS (i.e.,  $N_{ra} = 1$), while the BS sweeps through all its beamforming directions (i.e., $M_{ra} = M$) to receive the RA preamble sequences from the users. BS beamforming direction to the user is determined as the receive direction of the RA preamble. 

2) \emph{Fast RA}: cell search is the same as baseline, but the BS receives omni-directionally during random access. The user determines the BS beamforming direction  by decoding the beam reference signal during cell search, and it encodes that information into the RA preamble. The BS obtains its beamforming direction by decoding the RA preamble. Therefore,  $M_{cs} = M$, $N_{cs} = N$, $M_{ra} = 1$ (omni) and $N_{ra} = 1$ (beamforming in fixed direction). 

3) \emph{Fast CS}: in order reduce the cell search overhead while maintaining reasonable synchronization signal strength, the BS applies a coarse beam-sweeping using relatively wide beams during CS (i.e., $M_{cs} = m_{cs}$ with $N \leq m_{cs} \leq M$)~\cite{li2016performance}. Other procedures are the same as baseline.

4) \textit{Omni RX}: Now $M_{cs} = M$, $N_{cs} = 1$ (omni), $M_{ra} = 1$ (omni), $N_{ra} = N$, i.e., the user receives omni-directionally during cell search, and the BS receives omni-directionally during random access. The BS determines its beamforming direction as in the fast RA protocol, and the user determines its beamforming direction by beam sweeping during random access.

A summary of these protocols is provided in Table~\ref{IAOptions}. 

\begin{table}[t]
	\center\caption{Initial access protocols}\label{IAOptions}
	\resizebox{300pt}{!}{%
		\begin{tabular}{|c|p{40mm}|p{30mm}|p{30mm}|p{30mm}|}
			\hline 
			\textbf{Protocol} & \textbf{BS during CS} & \textbf{User during CS} & \textbf{BS during RA} & \textbf{User during RA}\\ 
			\hline 
			Baseline & Beam-sweeping ($M_{cs} = M$)& Beam-sweeping ($N_{cs} = N$) & Beam-sweeping ($M_{ra} = M$)& Fixed-direction ($N_{ra} = 1$)\\
			\hline
			Fast RA &  Beam-sweeping ($M_{cs} = M$) & Beam-sweeping ($N_{cs} = N$)  & Omni-directional ($M_{ra} = 1$) & Fixed-direction ($N_{ra} = 1$)  \\
			\hline
			Fast CS & Coarse beam-sweeping ($M_{cs} = m_{cs}$) & Beam-sweeping ($N_{cs} = N$) & Beam-sweeping ($M_{ra} = M$) & Fixed-direction ($N_{ra} = 1$) \\
			\hline
			UE Omni RX & Beam-sweeping ($M_{cs} = M$) & Omni-directional ($N_{cs} = 1$) & Omni-directional ($M_{ra} = 1$) & Beam-sweeping ($N_{ra} = N$)\\
			\hline
		\end{tabular} }
	\end{table}

\subsection{Performance Metrics}
The metrics that we use to evaluate the performance of the initial access protocols are:
\subsubsection{Success Probability of Initial Access} Initial access is successful if both cell search and random access are successful. For the typical user, we use $e_0$ and $\delta_0$ to denote its success indicator for cell search and random access in a typical initial access cycle. 
Therefore, the initial access success probability is given by: $\eta_{IA} = \mathbb{E}(e_0 \times \delta_0).$ 

\subsubsection{Expected Initial Access Delay}
If the user fails the initial access procedure in one initial access cycle, it will try to re-connect to the network in the next cycle. According to Fig.~\ref{fig:time_Structure}, the total initial access delay for typical user to be connected is given by:
\begin{align}\label{IA_Delay_Eq}
D_0 = (L_0 -1)  T + (M_{cs} N_{cs} \tau_{cs} + M_{ra} N_{ra} \tau_{ra}),
\end{align} 
where $L_0 \in \mathbb{N}^{+}$ represents the number of cycles to discover the typical user, $T$ represents the period of an initial access cycle, and $M_{cs} N_{cs} \tau_{cs} + M_{ra} N_{ra} \tau_{ra}$ represents the duration for initial access in each cycle. 
In this paper, we are focused on a high mobility scenario where the users or blockers (e.g., pedestrians and cars) are moving with a relatively high speed, such that the user and BS PPPs are independent across different initial access cycles. 
Therefore, $L_0$ follows a geometric distribution with parameter $\eta_{IA}$, which means:
\begin{align}\label{ExpDelayExpr}
\mathbb{E}(D_0) = (\frac{1}{\eta_{IA}} - 1) T +  (M_{cs} N_{cs} \tau_{cs} + M_{ra} N_{ra} \tau_{ra}).
\end{align}
Achieving a small initial access delay is important for mmWave cellular networks, especially for mobile scenarios where initial access needs to be performed very frequently. In addition, it is useful to reduce power consumption.


\subsubsection{Average User-Perceived Downlink Throughput} We only focus on the downlink and assume the entire data transmission period is occupied by the downlink in this paper. If the user succeeds the initial access, it can be scheduled by its serving BS for downlink transmission; otherwise, its data rate in the current cycle is 0 almost surely. 
Therefore, the average user-perceived downlink throughput, or average UPT, which represents the expected downlink data rate a typical user achieves within one initial access cycle, is given by:
\begin{align}\label{AvgUPTDefnEq}
\bar{R} = (1 -\eta_{TO}) \times \eta_{IA} \times \mathbb{E}[ \eta_{s}  W\log_{2}(1+\text{SINR}_{DL} )| e_0\delta_0 = 1],
\end{align}
where $\eta_{TO} \triangleq \min(\frac{M_{cs} N_{cs} \tau_{cs} + M_{ra} N_{ra} \tau_{ra}}{T}, 1)$ represents the initial access overhead; $\eta_{s}$ denotes the average schedule probability of typical user; and $\text{SINR}_{DL}$ denotes the downlink data SINR. 


It can be observed from~(\ref{ExpDelayExpr}) and~(\ref{AvgUPTDefnEq}) that both the initial access delay and UPT are highly dependent on $M_{cs},N_{cs},M_{ra}$, and $N_{ra}$, which means the four protocols in Table~\ref{IAOptions} could lead to very different initial access delay and UPT performance. In the rest of this paper, we develop and verify a general analytical framework that can quantify the impact of various initial access protocols on the mmWave system performance under a general blockage model.

\section{Success Probability for Cell Search and Random Access}\label{SPCSRASec}
In this section, the success probabilities for 
the initial access protocols in Table~\ref{IAOptions} are derived.
\vspace*{-0.1in}

\subsection{Success Probability for Cell Search}
\subsubsection{Analytical Model for Cell Search}\label{CS_Analy_Model_SubSec}
According to Section~\ref{CSSubSec}, BSs and users sweep through $M_{cs} \times N_{cs}$ transmit-receive beam pairs synchronously over the downlink during cell search. 
If the BS and user beam pair $(m,n)$ ($1 \leq m \leq M_{cs}$, $1 \leq n \leq N_{cs}$) is aligned, the typical user can receive from the main-lobes of BSs inside the following area due to the synchronous beam sweeping pattern:
\begin{align}
S(o,\frac{2\pi (n-1)}{N_{cs}}, \frac{2\pi n}{N_{cs}}) \cap S(o,\frac{2\pi (m-1)}{M_{cs}} + \pi, \frac{2\pi m}{M_{cs}} + \pi),
\end{align} 
where $o$ represents the typical user located at the origin of $\mathbb{R}^2$, and we define an infinite sector domain centered at $u \in \mathbb{R}^2$ by:  
\begin{align}
S(u,\theta_1,\theta_2) = \{ x \in \mathbb{R}^2, \text{ s.t., } \angle(x - u) \in [\theta_1, \theta_2) \}.
\end{align}From the typical user's perspective, there are $K_{cs} \triangleq \max(M_{cs},N_{cs})$ such non-overlapping sectors during cell search, which are referred to as the ``BS sectors" for the rest of the pape. Specifically, the $j$-th ($ 1 \leq j \leq K_{cs}$) BS sector is $S(o,\frac{2\pi (j-1)}{K_{cs}}, \frac{2\pi j}{K_{cs}})$. Note the BS sector notion in this paper is different from the sector concept in antenna theory. An example is shown in Fig.~\ref{fig:BF_Structure}.

For the rest of the paper, when analyzing the typical user performance inside a BS sector, we implicitly assume the BS and user beams are aligned. 
In addition, we say a BS sector is detected during cell search if the typical user is able to detect the BS that provides the smallest path loss in this sector, and the overall cell search is successful if at least one BS sector is detected. For simplicity, we neglect the scenario that the BS providing the smallest path loss inside a BS sector is in deep fade and unable to be detected, while some other BSs can be detected in the same sector. Such a scenario will be incorporated in our future work. After cell search, the typical user selects the BS with the smallest path loss across all the detected BS sectors as its serving BS, and initiates random access to this BS.

\subsubsection{Success Probability of Cell Search}
Since BSs are PPP and different BS sectors are non-overlapping, the event for BS sectors to be detected are independent and identically distributed (i.i.d.). Without loss of generality, we consider the first BS sector as a ``typical" BS sector, which is denoted by $S_1 \triangleq S(o,0,\frac{2\pi}{K_{cs}})$. 
The minimum path loss distribution inside the typical BS sector is given by the following lemma. 

\begin{lemma}\label{MinPLPerSectorLemma}
	Denote the minimum path loss from the typical user to BSs inside the typical BS sector by $Z_1$, then the probability density function (PDF) of $Z_1$ is given by:
	\allowdisplaybreaks
	\begin{align}\label{MinPLPerSectorEq}
	\allowdisplaybreaks
	f_{Z_1}(z) = &\biggl\{\frac{2\pi\lambda}{K_{cs}} \frac{1}{\alpha_L} (\frac{1}{\beta})^{\frac{2}{\alpha_L}} z^{\frac{2}{\alpha_L}-1} h((\frac{z}{\beta})^{\frac{1}{\alpha_L}}) \exp\biggl(-\frac{2\pi\lambda}{K_{cs}} \int_{0}^{(\frac{z}{\beta})^{\frac{1}{\alpha_L}}} \!\!\!\!\!h(r) r {\rm d}r\biggl) \biggl\} \nonumber \\
	&\times \exp\biggl(-\frac{2\pi\lambda}{K_{cs}} \int_{0}^{(\frac{z}{\beta})^{\frac{1}{\alpha_N}}} \!\!\!\!\!\!(1-h(r)) r {\rm d}r\biggl) + \biggl\{\frac{2\pi\lambda}{K_{cs}} \frac{1}{\alpha_N} (\frac{1}{\beta})^{\frac{2}{\alpha_N}} z^{\frac{2}{\alpha_N}-1}(1 - h((\frac{z}{\beta})^{\frac{1}{\alpha_N}}))  \nonumber\\
	&\times \exp\biggl(-\frac{2\pi\lambda}{K_{cs}} \int_{0}^{(\frac{z}{\beta})^{\frac{1}{\alpha_N}}} \!\!\!\!\!\!(1-h(r)) r {\rm d}r\biggl) \biggl\}\exp\biggl(-\frac{2\pi\lambda}{K_{cs}} \int_{0}^{(\frac{z}{\beta})^{\frac{1}{\alpha_L}}} \!\!\!\!\!h(r) r {\rm d}r\biggl).
	\end{align} where $K_{cs} = \min(M_{cs},N_{cs})$.
\end{lemma}

\begin{proof}
Inside the typical BS sector, since the minimum path loss to the typical user is either from the nearest LOS BS or the nearest NLOS BS, we have: 
\begin{align*}
\mathbb{P}(Z_1 \geq z) 
& \overset{(a)}{=} \mathbbm{P}\biggl(\min_{x \in \Phi_L \cap S_1}\|x\| \geq  \mathit{l}_L^{-1}(z)\biggl) \times \mathbb{P}\biggl( \min_{x \in \Phi_N \cap S_1} \|x\| \geq \mathit{l}_N^{-1}(z)\biggl)\\
&\overset{(b)}{=} \exp\biggl(-\frac{2\pi\lambda}{K_{cs}} \int_{0}^{(\frac{z}{\beta_L})^{\frac{1}{\alpha_L}}} \!\!\!\!\!\!\!\! h(r) r {\rm d}r\biggl) \exp\biggl(-\frac{2\pi\lambda}{K_{cs}} \int_{0}^{(\frac{z}{\beta_N})^{\frac{1}{\alpha_N}}} \!\!\!\biggl(1-h(r)\biggl) r {\rm d}r\biggl),
\end{align*}
where (a) is because $\Phi_L$ and $\Phi_N$ are independent, and (b) is from void probability of PPP.
\end{proof}

Note the first term and second term in (\ref{MinPLPerSectorEq}) refer to the PDF of $Z_1$ when the BS providing the minimum path loss is LOS and NLOS respectively. 

\begin{remark}
	The result in  (\ref{MinPLPerSectorEq}) is in integral form since it provides the path loss distribution under a general blockage model. It can be simplified for specific blockage models such as the LOS ball model and exponential blockage model. For example, for the LOS ball model with $p =1$, Lemma~\ref{MinPLPerSectorEq} simply becomes:
	\begin{align*}
	f_{Z_1}(z) = &\frac{2\pi\lambda}{K_{cs}} \frac{1}{\alpha_L} (\frac{1}{\beta})^{\frac{2}{\alpha_L}} z^{\frac{2}{\alpha_L}-1} \exp\biggl(-\frac{\pi\lambda }{K_{cs}} (\frac{z}{\beta})^{\frac{1}{\alpha_L}}\biggl) \mathbbm{1}_{(\frac{z}{\beta})^{\frac{1}{\alpha_L}} \leq R_c}\\
	& + \frac{2\pi\lambda}{K_{cs}} \frac{1}{\alpha_N} (\frac{1}{\beta})^{\frac{2}{\alpha_N}} z^{\frac{2}{\alpha_N}-1} \exp\biggl(-\frac{\pi\lambda}{K_{cs}} (\frac{z}{\beta})^{\frac{1}{\alpha_N}}\biggl) \mathbbm{1}_{(\frac{z}{\beta})^{\frac{1}{\alpha_N}} \geq R_c}.
	\end{align*}
	Similarly, all the analytical results afterwards can be simplified under specific blockage models.
\end{remark}

Conditionally on the minimum path loss inside the typical BS sector $Z_1$, the SINR of the synchronization signal from the BS providing the minimum path loss is given by: 
\allowdisplaybreaks
\begin{align}\label{SyncSigSINREq}
\allowdisplaybreaks
\text{SINR}_{SS}(Z_1) = \frac{F_0/Z_1}{\!\!\!\!\!\!\!\!\!\!\sum\limits_{x_{i}^L \in \Phi_L \cap S_1 \cap B^c(o,\mathit{l}_L^{-1}(Z_1))} \!\!\!\!\!\!\!\!\!\!\!\!\!\! F_{i}^L / \mathit{l}_{L}(\|x_{i}^L\|) + \!\!\!\!\!\!\!\!\!\!\!\!\! \sum\limits_{x_{j}^N \in \Phi_N \cap S_1 \cap B^c(o,\mathit{l}_N^{-1}(Z_1))} \!\!\!\!\!\!\!\!\!\! F_{j}^N / \mathit{l}_{N}(\|x_{j}^N\|) + \frac{\sigma^2}{P_b M_{cs} G( 2\pi /N_{cs})}},
\end{align}
where $F_0$, $F_i^L$ and $F_j^N$ represent the Rayleigh fading channel from the typical user to the BS proving the minimum path loss, interfering LOS BS $x_{i}^L$ and interfering NLOS BS $x_{j}^N$ respectively. The last term in the denominator of~(\ref{SyncSigSINREq}) represents the ``effective noise" at the typical user, which is the total noise power normalized by the transmit power and antenna gains. In particular, the user antenna gain is $G(\frac{2\pi}{N_{cs}})$, and the BS antenna gain is $M_{cs}$ since it has 0 side lobe gain. We have applied the strong Markov property of PPPs~\cite[Proposition 1.5.3]{baccelli2010stochastic} for obtaining~(\ref{SyncSigSINREq}): conditionally on the minimum path loss $Z_1$, the interference only depends on the interfering LOS and NLOS BSs located inside $S_1 \cap B^c(0,\mathit{l}_L^{-1}(Z_1))$ and $S_1 \cap B^c(0,\mathit{l}_N^{-1}(Z_1))$ respectively. 
The detection probability of the typical BS sector is as follows:
\begin{lemma}\label{CSSucProbPerSectorLemma}
	The probability for the typical user to detect the BS providing the smallest path loss inside the typical BS sector is given by:
	\allowdisplaybreaks
	\begin{align}\label{CSProbPerSecEq}
	\allowdisplaybreaks
		P_{M_{cs},N_{cs}} (\Gamma_{cs}) = \int_{0}^{\infty}	\tilde{P}_{M_{cs},N_{cs}}(z,\Gamma_{cs}) f_{Z_1}(z) {\rm d}z,
	\end{align}
where $\tilde{P}_{M_{cs},N_{cs}}(z,\Gamma_{cs})$ denotes the conditional detection probability when the minimum path loss inside the typical BS sector is $z$, which is given by:
	\begin{align}\label{GFuncDefnEq}
	\allowdisplaybreaks
	\tilde{P}_{M_{cs},N_{cs}} (z, \Gamma_{cs}) = &\exp\biggl(-\frac{\Gamma_{cs}z\sigma^2}{P_b M_{cs} G(2\pi/N_{cs})}\biggl) V\left(z, \Gamma_{cs},\frac{\lambda }{K_{cs}}\right),
	\end{align}
	where function $V$ is defined in~(\ref{SpecFnDefnEq}), and other parameters are defined in Table~\ref{SysParaTable}.
\end{lemma}

\begin{proof}
	Given the minimum path loss inside the typical BS sector is $Z_1 = z$, the conditional success probability to detect the BS providing the minimum path loss is:
	\allowdisplaybreaks
	\begin{align*}
	\allowdisplaybreaks
	\!\!\!&\mathbb{P}(\text{SINR}_{SS}(z) > \Gamma_{cs}) \nonumber\\
	\overset{(a)}{=}&\exp\biggl(-\frac{\Gamma_{cs}z\sigma^2}{PM_{cs}G(2\pi/N_{cs})}\biggl) \mathbb{E}\biggl[\exp(-\Gamma_{cs}z \!\!\!\!\!\!\!\!\!\!\!\!\!\!\!\!\!\sum\limits_{x_{i}^L \in \Phi_L \cap S_1 \cap B^c(o,\mathit{l}_L^{-1}(z))} \!\!\!\!\!\!\!\!\!\!\!\!\!\!\!\!\!F_{i}^L / \mathit{l}_{L}(x_{i}^L))\biggl] \mathbb{E}\biggl[\exp(-\Gamma_{cs}z \!\!\!\!\!\!\!\!\!\!\!\!\!\!\!\!\!\sum\limits_{x_{j}^N \in \Phi_N \cap S_1 \cap B^c(o,\mathit{l}_N^{-1}(z))} \!\!\!\!\!\!\!\!\!\!\!\!\!\!\!\!\!F_{j}^N / \mathit{l}_{N}(x_{j}^N))\biggl] \nonumber\\
	\overset{(b)}{=}&\exp\biggl(-\frac{\Gamma_{cs}z\sigma^2}{PM_{cs}G(2\pi/N_{cs})}\biggl) \mathbb{E}\biggl[\!\!\!\!\!\!\!\!\!\!\!\!\!\prod_{{x_{i}^L \in \Phi_L \cap S_1 \cap B^c(o,\mathit{l}_L^{-1}(z))}}\!\!\!\!\!\frac{\mathit{l}_L(x_i^L)}{\mathit{l}_L(x_i^L)+\Gamma_{cs}z}\biggl] \mathbb{E}\biggl[\!\!\!\!\!\!\!\!\prod_{{x_{j}^N \in \Phi_N \cap S_1 \cap B^c(o,\mathit{l}_N^{-1}(z))}}\!\!\!\frac{\mathit{l}_N(x_j^N)}{\mathit{l}_N(x_j^N)+\Gamma_{cs}z}\biggl],
	\end{align*} 
	where (a) is from the expression of $\text{SINR}_{SS}(z)$ in~(\ref{SyncSigSINREq}), and (b) is because all the fadings variables are i.i.d. exponentially distributed with parameter 1. Therefore, $\tilde{P}_{M_{cs},N_{cs}} (z, \Gamma_{cs})$ can be obtained by applying the probability generating functional of the PPP~\cite{chiu2013stochastic}. Finally, the overall detection probability of the typical BS sector is obtained by de-conditioning on $z$.
\end{proof}

We can derive the overall cell search success probability as follows:
\begin{theorem}\label{CS_SP_Thm}
	The probability for the typical user to succeed the cell search is given by: 
	\allowdisplaybreaks
	\begin{align}\label{CS_SP_Thm_Eq}
	\allowdisplaybreaks
	\hat{P}_{M_{cs},N_{cs}}(\Gamma_{cs}) = 1-(1-P_{M_{cs},N_{cs}}(\Gamma_{cs}))^{K_{cs}},
	\end{align}
	where $P_{M_{cs},N_{cs}}(\Gamma_{cs})$ is derived in Lemma~\ref{CSSucProbPerSectorLemma}. 
\end{theorem}

\begin{proof}
	Since the BS process is PPP and all the BS sectors are non-overlapping, every BS sector can be detected by the typical user independently with success probability $P_{M_{cs},N_{cs}}(\Gamma_{cs})$. The proof is concluded by noting cell search is successful if the typical user is able to detect at least one BS sector.
\end{proof}

\begin{remark}
	Intuitively, by increasing the number of beamforming directions to search (i.e., $K_{cs}$), the synchronization signal received at the typical user is subject to less effective noise as well as less interference on average, and therefore a higher cell search success probability is expected. This observation will be validated more rigorously in Section~\ref{NumEvalSec}. 
\end{remark}

\subsubsection{Serving Path Loss Distribution} \label{ServingPLDistrSubSec}
Since the main objective of cell search is for the typical user to detect its neighboring BSs and make cell association decision, it is important to determine the path loss distribution from the typical user to its potential serving BS. For the rest of this paper, we call the potential serving BS of the typical user the ``tagged BS". 

Denote by $Z_0$ the path loss from the typical user to the tagged BS; it is the minimum path loss from the typical user to the BSs inside the detected BS sectors. By convention if cell search fails, we say that the potential serving BS to the typical user is infinitely far away and therefore $Z_0$ is infinity. Based on Lemma~\ref{CSSucProbPerSectorLemma}, we are able to derive the distribution of $Z_0$ as follows:
\begin{lemma}\label{minPLdistlemma}
	The complementary cumulative distribution function (CCDF) of the path loss from the typical user to the tagged BS is given by:
	\allowdisplaybreaks
	\begin{align}\label{minPLdistEq}
	\allowdisplaybreaks
	\mathbb{P}(Z_0 \geq z_0 ) = \biggl(P_{M_{cs},N_{cs}}(z_0,\Gamma_{cs}) + 1- P_{M_{cs},N_{cs}}(\Gamma_{cs})\biggl)^{K_{cs}},
	\end{align}
	where $P_{M_{cs},N_{cs}}(z_0,\Gamma_{cs}) \triangleq  \int_{z_0}^{\infty} \tilde{P}_{M_{cs},N_{cs}}(z,\Gamma_{cs}) f_{Z_1}(z) {\rm d}z$. In addition, the PDF of $Z_0$ is given by:
	\begin{align*}
	f_{Z_0}(z_0) = K_{cs}\biggl(P_{M_{cs},N_{cs}}(z_0,\Gamma_{cs}) + 1- P_{M_{cs},N_{cs}}(\Gamma_{cs})\biggl)^{K_{cs}-1} \!\! \tilde{P}_{M_{cs},N_{cs}}(z_0,\Gamma_{cs}) f_{Z_1}(z_0),
	\end{align*}	
	where the notation and functions are defined in Table~\ref{SysParaTable}.
\end{lemma}
\begin{proof}
	 Note that if the typical user is unable to detect a certain BS sector, the BSs inside this sector can be seen as infinitely far away from the typical user, or equivalently having an infinite path loss. Therefore, 
	 $Z_0 \geq z_0$ is equivalent to the fact that for any BS sector, either this sector cannot be detected, or this sector is detected and the minimum path loss from the typical user to BSs inside this sector is greater than or equal to $z_0$. For the typical BS sector, the above events happen with the following probability:
	 \allowdisplaybreaks
	 \begin{align}\label{MinPL_post_CS_Proof_Eq1}
	 \allowdisplaybreaks
	 &\mathbb{P}(Z_1 \geq z_0 | \text{SINR}_{SS}(Z_1) \geq \Gamma_{cs}) \times \mathbb{P}(\text{SINR}_{SS}(Z_1) \geq \Gamma_{cs}) +  \mathbb{P}(\text{SINR}_{SS}(Z_1) < \Gamma_{cs})\nonumber\\
	 =&\mathbb{P}(Z_1 \geq z_0 \cap \text{SINR}_{SS}(Z_1) \geq \Gamma_{cs}) + 1-\mathbb{P}(\text{SINR}_{SS}(Z_1) \geq \Gamma_{cs})]\nonumber\\
	 =&\int_{0}^{\infty} \mathbbm{1}_{z_1 \geq z_0} \mathbb{P}(\text{SINR}_{SS}(z_1) \geq \Gamma_{cs}) f_{Z_1}(z_1){\rm d}z_1 + 1-P_{M_{cs},N_{cs}}(\Gamma_{cs})\nonumber\\
	 =&P_{M_{cs},N_{cs}}(z_0,\Gamma_{cs}) + 1-P_{M_{cs},N_{cs}}(\Gamma_{cs}),
	 \end{align}
	where $Z_1$ denotes the minimum path loss from typical user to BSs inside the typical BS sector. Finally, we can obtain~(\ref{minPLdistEq}) since the detection events for the BS sectors are independent from each other, and $Z_0 \geq z_0$ is equivalent to~(\ref{MinPL_post_CS_Proof_Eq1}) is satisfied by all BS sectors.
\end{proof}

\begin{remark}
	It is straightforward that $\lim\limits_{z_0 \rightarrow \infty} \mathbb{P}(Z_0 \geq z_0) = 1- \hat{P}_{M_{cs},N_{cs}}(\Gamma_{cs})$, which means the probability for $Z_0$ to have a mass at infinity is equal to the probability that typical user fails cell search. Since $\hat{P}_{M_{cs},N_{cs}}(\Gamma_{cs})$ is 
	non-decreasing w.r.t. $K_{cs}$, $Z_0$ has a lighter tail as $K_{cs}$ increases. 
\end{remark}

\subsection{Success Probability for Random Access}\label{RA_SP_SubSec}
According to Section~\ref{CSSubSec}, users that succeed cell search will initiate the random access procedure, where BSs and users sweep through $M_{ra} \times N_{ra}$ transmit-receive beam pairs over the uplink synchronously.  
Since each user can initiate random access only upon successful cell search, the users that are involved in the random access process has intensity $\lambda_u \hat{P}_{M_{cs},N_{cs}}(\Gamma_{cs})$. For analytical tractability, we assume the user process during RA is approximated by a homogeneous PPP with intensity $\lambda_u \hat{P}_{M_{cs},N_{cs}}(\Gamma_{cs})$, and we will show in Remark~\ref{ESFRmk} and Section~\ref{IANumSubSec} that this is a reasonable approximation. 
Since random access is successful if the RA preamble of the typical user can be decoded by the tagged BS without any collision, the success probability for random access is derived in the following two parts.

\subsubsection{No RA Preamble Collision Probability} The RA preamble collision happens at the typical user when there exists another user such that: (1) it tries to associate with the tagged BS; (2) it chooses the same RA preamble sequence as the typical user, and (3) the tagged BS receives the RA preamble from this user under the same receive beam as the typical user. Therefore, the probability that the typical user has no RA premable collision is as follows:
\begin{lemma}\label{No_PA_Collision_Lemma}
	The probability that the typical user is not subject to RA preamble collision can be approximated by:
	\allowdisplaybreaks
	\begin{align}\label{Prob_No_PA_Collision_Eq}
	P_{co} \approx \exp(-\frac{1.28\lambda_u \hat{P}_{M_{cs},N_{cs}}(\Gamma_{cs})}{\lambda N_{PA} M_{ra}}),
	\end{align}
	where $\hat{P}_{M_{cs},N_{cs}}(\Gamma_{cs})$ is derived in Theorem~\ref{CS_SP_Thm}, and other parameters are defined in Table~\ref{SysParaTable}.
\end{lemma}

\begin{proof}
	Since the association from the typical user to the tagged BS is stationary~\cite{singh2014onassoication}, the mean associated cell size of the tagged BS is given by: $ \frac{1.28}{\lambda}$. In particular, the factor of ``1.28" is due to the fact that the association cell of the tagged BS is an area-biased version to that of a typical BS~\cite{singh2014onassoication}, whose accuracy has been verified in~\cite{single2013offloading}. Since each user randomly chooses its RA preamble sequence out of $N_{PA}$ total sequences, and the receive beamwidth of the tagged BS is $\frac{2\pi}{M_{ra}}$, the probability that a user associated with the tagged BS collides with the typical user for random access is $\frac{1}{N_{PA} M_{ra}}$. As will be demonstrated in Remark~\ref{ESFRmk}, the user process during random access can be accurately modeled by a stationary PPP with intensity $\lambda_u \hat{P}_{M_{cs},N_{cs}(\Gamma_{cs})}$. Therefore, the proof can be concluded from the void probability of the PPP~\cite{chiu2013stochastic}.
\end{proof}

\subsubsection{Successful Reception Probability of RA Preamble} The RA preamble sequence of the typical user is successfully decoded if its received SINR at the tagged BS is greater than or equal to $\Gamma_{ra}$. 
For simplicity, we assume perfect RA preamble sequences are used, so that they have a delta function as their auto-correlation functions and zero as their cross-correlation functions. Thus, only the users choosing the same RA preamble sequence as the typical user can potentially interferer with it. Conditionally on the path loss from the typical user to the tagged BS, the successful reception probability of the RA preamble is as follows:
\begin{lemma}\label{RA_SINR_COP_Lemma}
 	Denote by $Z_0$ the path loss from the typical user to the tagged BS, the probability that the RA preamble of the typical user can be successfully received by the tagged BS is: 
 	\allowdisplaybreaks
 	\begin{align}\label{RA_SINR_COP_Lemma_Eq}
 	\allowdisplaybreaks
 	P_{ra} (Z_0,\Gamma_{ra}) =& \exp\biggl(-\frac{\Gamma_{ra} Z_0 \sigma^2}{P_uM_{ra} G(2\pi/N)}\biggl) U\biggl(Z_0,\Gamma_{ra},\frac{\lambda_u \hat{P}_{M_{cs},N_{cs}}(\Gamma_{cs})}{N M_{ra} N_{PA}}\biggl) \nonumber\\
 	 &\times U\biggl(Z_0,\frac{g(2\pi/N)}{G(2\pi/N)}\Gamma_{ra},(1-\frac{1}{N})\frac{\lambda_u \hat{P}_{M_{cs},N_{cs}}(\Gamma_{cs})}{M_{ra} N_{PA}}\biggl) ,
 	\end{align}
 	where $N = \max(N_{cs},N_{ra})$, and $U$ is defined in~(\ref{SpecFnDefnEq}).
\end{lemma}
\begin{proof}
The proof is provided in Appendix~\ref{RA_SINR_COP_Lemma_Proof_Apdx}.
\end{proof}

Since $P_{ra} (Z_0,\Gamma_{ra}) = 0$ when $Z_0 = \infty$, 
the overall success probability of the initial access procedure can be obtained by combining Lemma~\ref{minPLdistlemma}, Lemma~\ref{No_PA_Collision_Lemma} and Lemma~\ref{RA_SINR_COP_Lemma}, which gives:

\begin{theorem}\label{Overall_IA_SP_Thm}
	The initial access success probability for the typical user is given by:
	\allowdisplaybreaks
	\begin{align}\label{IA_SP_Eq}
	\allowdisplaybreaks
	\eta_{IA} = &\int_{0}^{\infty}K_{cs}\biggl(P_{M_{cs},N_{cs}}(z_0,\Gamma_{cs}) + 1- P_{M_{cs},N_{cs}}(\Gamma_{cs})\biggl)^{K_{cs}-1} \nonumber \\
	 & \times \tilde{P}_{M_{cs},N_{cs}}(z_0,\Gamma_{cs}) \times P_{co}  \times P_{ra} (z_0,\Gamma_{ra})  f_{Z_1}(z_0) {\rm d}z_0,
	\end{align}
	where the notation and functions are defined in Table~\ref{SysParaTable}. 
\end{theorem}


\begin{remark}\label{ESFRmk}
Denote by $\Phi_u^{''}$ the users that succeed initial access, Theorem~\ref{Overall_IA_SP_Thm} shows the intensity of $\Phi_u^{''}$ is $\lambda_u \eta_{IA}$. Intuitively, $\Phi_u^{''}$ is expected to exhibit spatial clustering since the users in $\Phi_u^{''}$ should be centered around BSs and sparse at cell edges. In Fig.~\ref{ESF}, we plot the empty space function (ESF) of $\Phi_u^{''}$ which is defined as $F(r) \triangleq \mathbb{P}_{\Phi}^{0}\left( \min\{\|u\| : u \in \Phi_u^{''} \} \leq r \right)$, where $\mathbb{P}_{\Phi}^{0}$ denotes the Palm distribution of BS process $\Phi$. Fig.~\ref{ESF} shows that $\Phi_u^{''}$ has a smaller ESF than its fitted PPP, which means $\Phi_u^{''}$ exhibits clustered pattern~\cite{li2015statistical}. In fact, Fig.~\ref{ESF} also shows that for most range of $r$, the ESF of $\Phi_u^{''}$ falls within the 95\% confidence interval created by its fitted PPP. Therefore, we still assume $\Phi_u^{''}$ is modeled by a PPP with intensity $\lambda_u \eta_{IA}$ for analytical simplicity. The accuracy of this assumption will be validated in Section~\ref{NumEvalSec}. 
\end{remark}

\begin{figure}[h]
	\begin{subfigure}[b]{0.46\textwidth}
		\centering
		\includegraphics[height=1.8in, width= 2.9in]{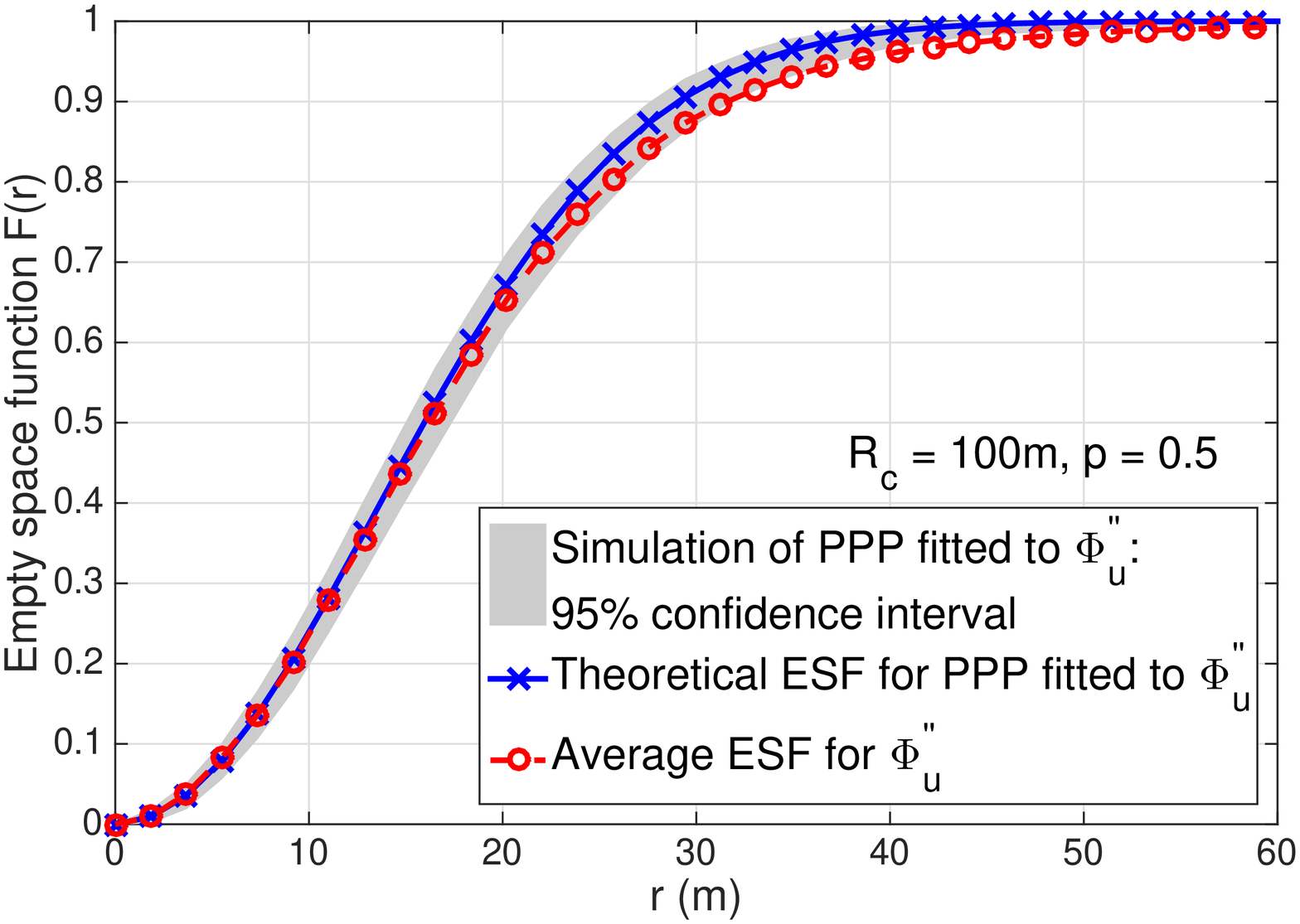}
		\caption{Generalized LOS ball model }
		\label{ESF_LOS}
	\end{subfigure}
	\hfill
	\begin{subfigure}[b]{0.46\textwidth}
		\centering
		\includegraphics[height=1.8in, width=2.9in]{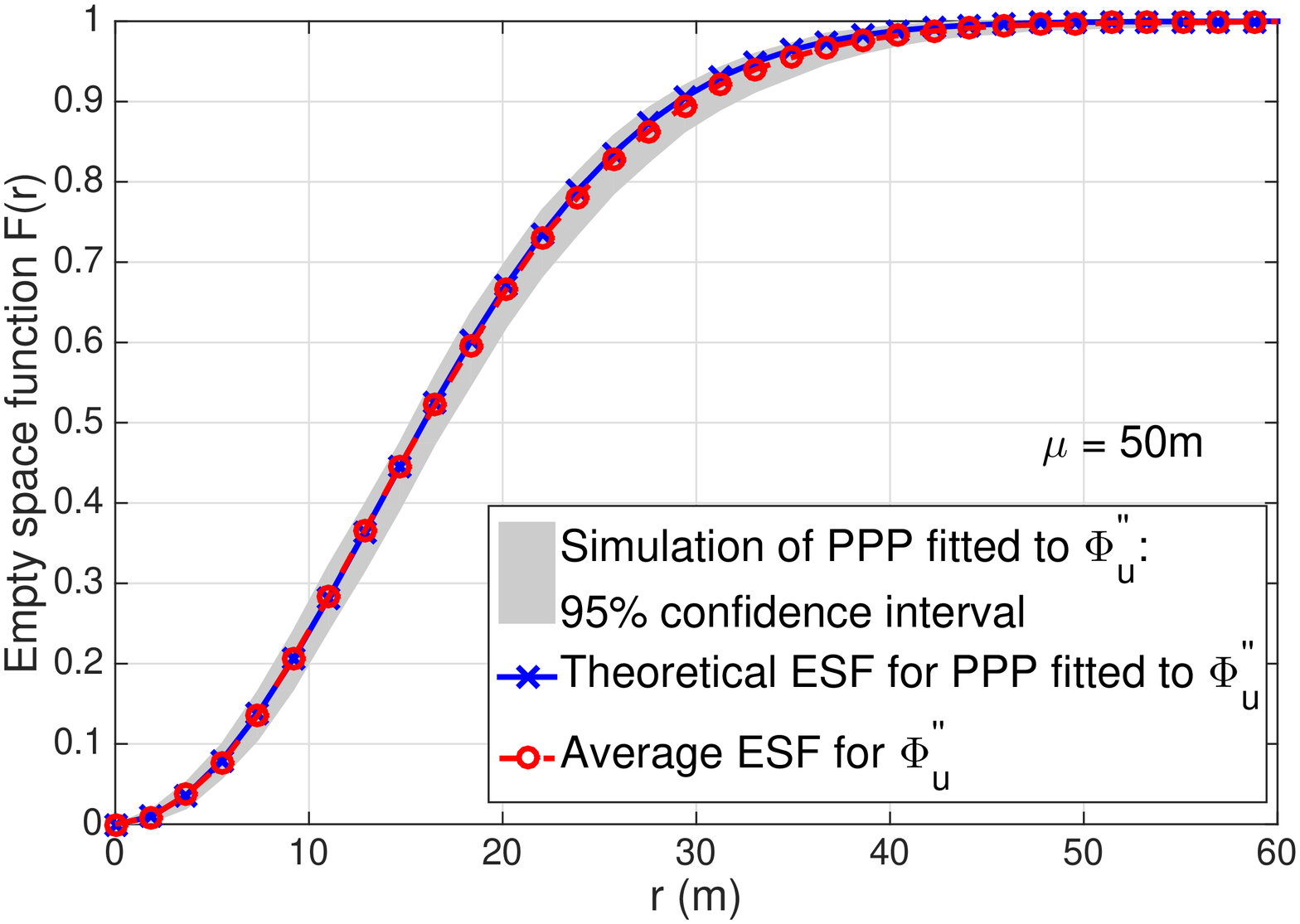}
		\caption{Exponential blockage model }
		\label{ESF_Exp}
	\end{subfigure}
	\caption{Empty space function comparison of $\Phi_u^{''}$ and its fitted PPP ($M = 8$, $N =4$).}\label{ESF}
\end{figure}

Based on Theorem~\ref{Overall_IA_SP_Thm}, the expected initial access delay defined in~(\ref{ExpDelayExpr}) can be easily evaluated, which will be discussed in more detail in Section~\ref{NumEvalSec}. 

\section{Downlink SINR Distribution and User-perceived Throughput}\label{DLUPTSect}
In this section, the downlink SINR distribution and the user-perceived throughput are derived. 

During the data transmission phase in Fig.~\ref{fig:time_Structure}, each BS randomly schedules one of its associated users, and the beam directions of the BS and its scheduled user are aligned. 
We assume the typical user has succeeded initial access and is scheduled by the tagged BS, such that the CCDF of its conditional downlink data SINR is given by:
\begin{align}\label{DL_SINR_Exp_Eq}
P_{DL}(\Gamma) = \mathbb{P}(\text{SINR}_{DL} \geq \Gamma | e_0\delta_0 = 1), 
\end{align}
where $\text{SINR}_{DL}$ represents the SINR of the typical user, while $e_0$ and $\delta_0$ represent the success indicator for cell search and random access respectively. 

Since a random scheduler is used, we assume i.i.d. beam directions of the interfering BSs to the typical user. Despite every interfering BS has positive probability to have zero associated users, we assume it is actively transmitting for analytical simplicity. Although this overestimates the interference at the typical user, we will show in Section~\ref{NumEvalSec} that the effect is negligible.  
Based on the assumptions above, the expression of $P_{DL}(\Gamma)$ is derived in the following lemma:
\begin{lemma}\label{CCDF_Data_SINR_Lemma}
	The CCDF of the SINR of the typical user given it succeeds the initial access is approximated by:
	\allowdisplaybreaks
	\begin{align*}
	\allowdisplaybreaks
	P_{DL}(\Gamma) =& \frac{1}{\eta_{IA}}\int_{0}^{\infty} \exp\biggl(-\frac{\Gamma z\sigma^2}{P_bMN}\biggl) K_{cs} \biggl[V(z,\Gamma,\frac{\lambda}{M K_{cs}})P_{M_{cs},N_{cs}}(z,\Gamma_{cs}) + U(z,\Gamma,\frac{\lambda}{M K_{cs}}) \nonumber \\
	&\times   (1 - P_{M_{cs},N_{cs}}(\Gamma_{cs})) \biggl]^{q-1} \biggl[ P_{M_{cs},N_{cs}}(z,\Gamma_{cs}) V(z,\frac{g(2\pi/N)}{G(2\pi/N)}  \Gamma,\frac{\lambda}{M K_{cs}})  +  (1 - P_{M_{cs},N_{cs}}(\Gamma_{cs})) \nonumber \\
	& \times U(z,\frac{g(2\pi/N)}{G(2\pi/N)}  \Gamma,\frac{\lambda}{M K_{cs}})\biggl]^{K_{cs}-q} \!\!\!\!  V(z,\Gamma,\frac{\lambda}{M K_{cs}}) \tilde{P}_{M_{cs},N_{cs}} (z, \Gamma_{cs}) P_{ra} (z,\Gamma_{ra}) P_{co} f_{Z_1}(z) {\rm d}z,
	\end{align*}
	where $q = \frac{K_{cs}}{N}$, and other notation and functions are all defined in Table~\ref{SysParaTable}. 
\end{lemma}
\begin{proof}
The proof is provided in Appendix~\ref{CCDF_Data_SINR_Proof_Apdx}.
\end{proof}

Given the data SINR distribution, we are able to derive the average user-perceived downlink throughput defined in~(\ref{AvgUPTDefnEq}). According to Remark~\ref{ESFRmk}, the users that succeed initial access are assumed to form a homogeneous PPP with intensity $\lambda_u \eta_{IA}$. 
As a result, the average number of users that are associated to the tagged BS is approximated by $1 + 1.28 \frac{\lambda_u \eta_{IA}}{\lambda}$, which means the average scheduling probability for the typical user is $\eta_s = \frac{1}{1 +  1.28 \lambda_u \eta_{IA} / \lambda}$. By substituting $\eta_s$ into~(\ref{AvgUPTDefnEq}), we can derive the average user-perceived downlink throughput as follows:

\begin{theorem}\label{Avg_UPT_thm}
	The average user-perceived downlink throughput is given by:
	\begin{align}\label{Avg_UPT_thm_eq} 
	 \bar{R} = \max(0,1-\frac{M_{cs}N_{cs} \tau_{cs} + M_{ra}N_{ra}\tau_{ra} }{T}) \times \frac{\eta_{IA}} {1 + 1.28 \lambda_u \eta_{IA} / \lambda} \times \int_{0}^{\infty} \frac{W}{\ln 2}\frac{P_{DL}(\Gamma)  {\rm d}\Gamma }{1+\Gamma},
	\end{align}
	where $P_{DL}(\Gamma)$ is derived in Lemma~\ref{CCDF_Data_SINR_Lemma}, and other notations are defined in Table~\ref{SysParaTable}.
\end{theorem}

\section{Numerical Evaluation for the Baseline Initial Access Protocol}\label{NumEvalSec}
Since the baseline protocol is the most straightforward initial access design which could be potentially implemented by the initial mmWave systems~\cite{Verizon20165G2}, a detailed performance evaluation is carried out in this section.

We consider a mmWave cellular system with the same frame structure and synchronization signal configuration as the one specified in~\cite{Verizon20165G2}. Specifically, the system operates at 28 GHz carrier frequency with 100 MHz bandwidth, the sub-carrier spacing is 75 kHz, and the corresponding OFDM symbol length (including cyclic prefix) is 14.3 $\mu$s. Each synchronization signal occupies only one OFDM symbol (i.e., $\tau_{cs} = 14.3$ $\mu$s), and the beam reference signal is also transmitted in the same symbol to uniquely identify the beam index. The synchronization signal/beam reference signal transmission period is 20 ms, which means $T$ = 20 ms. In addition, we assume each RA preamble sequence duration is also one OFDM symbol, and therefore $\tau_{ra} = 14.3$ $\mu$s. The default system parameter values are summarized in Table~\ref{SysParaTable}. 

In order to simulate the initial access and data transmission procedures, we have generated 50 realizations of the BS PPP, and 50 realizations of the user PPP given every BS PPP, inside a 1.5 km $\times$ 1.5 km network area. For each pair of the BS and user PPPs, we first simulate the initial access procedure according to Section~\ref{IASect}. Then the downlink data transmission phase is simulated, where each BS either randomly schedules one of its associated users, or keeps silent if it has no user to serve. By averaging over all the 2500 combinations of BS and user PPPs, different performance metrics of this mmWave system are recorded. 
BS and user locations are simulated by PPPs since currently there is no location data for mmWave system, and PPPs have already been shown to be an accurate model for mmWave system design~\cite{akoum2012covrage,singh2015tractable,Jihong2016Tractable,bai2015coverage,alkhateeb2016initial,DiRenzo2015Stochastic}.
 
\subsection{Performance for the Initial Access Phase}\label{IANumSubSec}
\subsubsection{Success Probability for Cell Search}
The cell search success probability is plotted in Fig.~\ref{SP_CS_Fig} for the generalized LOS ball model and the exponential blockage model. It can be observed from Fig.~\ref{SP_CS_Fig} that the analytical result in Theorem~\ref{CS_SP_Thm} is accurate. In addition, Fig.~\ref{SP_CS_Fig} shows that when BSs transmit omni-directionally and users receive omni-directionally, the cell search success probability is relatively low for various $\Gamma_{cs}$, which means the system is subject to significant coverage issues when cell search is performed omni-directionally as LTE. For example when $\Gamma_{cs} = -4$dB, Fig.~\ref{SP_CS_Fig} shows that the omni-directional cell search success probability is less tan 75\%. By contrast, it has been shown in Fig. 5 of~\cite{trac} that the omni-directional cell search success probability for macro cellular networks in lower frequencies is close to 90\% under the regular BS location model. 

Fig.~\ref{SP_CS_Fig} also shows when beam sweeping is applied, the cell search success probability can be significantly improved even with a small value of $\max(M_{cs},N_{cs})$ such as 4. As $M_{cs}$ or $N_{cs}$ is increased, the cell search probability can be further improved, so beam sweeping needs to be applied to guarantee a reasonable cell search performance. In the remaining simulations, we use $\Gamma_{cs} = -4$ dB as the SINR threshold to detect the synchronization signals, above which a sufficiently small miss detection probability (e.g., 1\%) can be achieved~\cite{barati2015directional}.

\begin{figure}[h]
	\begin{subfigure}[b]{0.46\textwidth}
		\centering
		\includegraphics[height=1.7in, width= 2.6in]{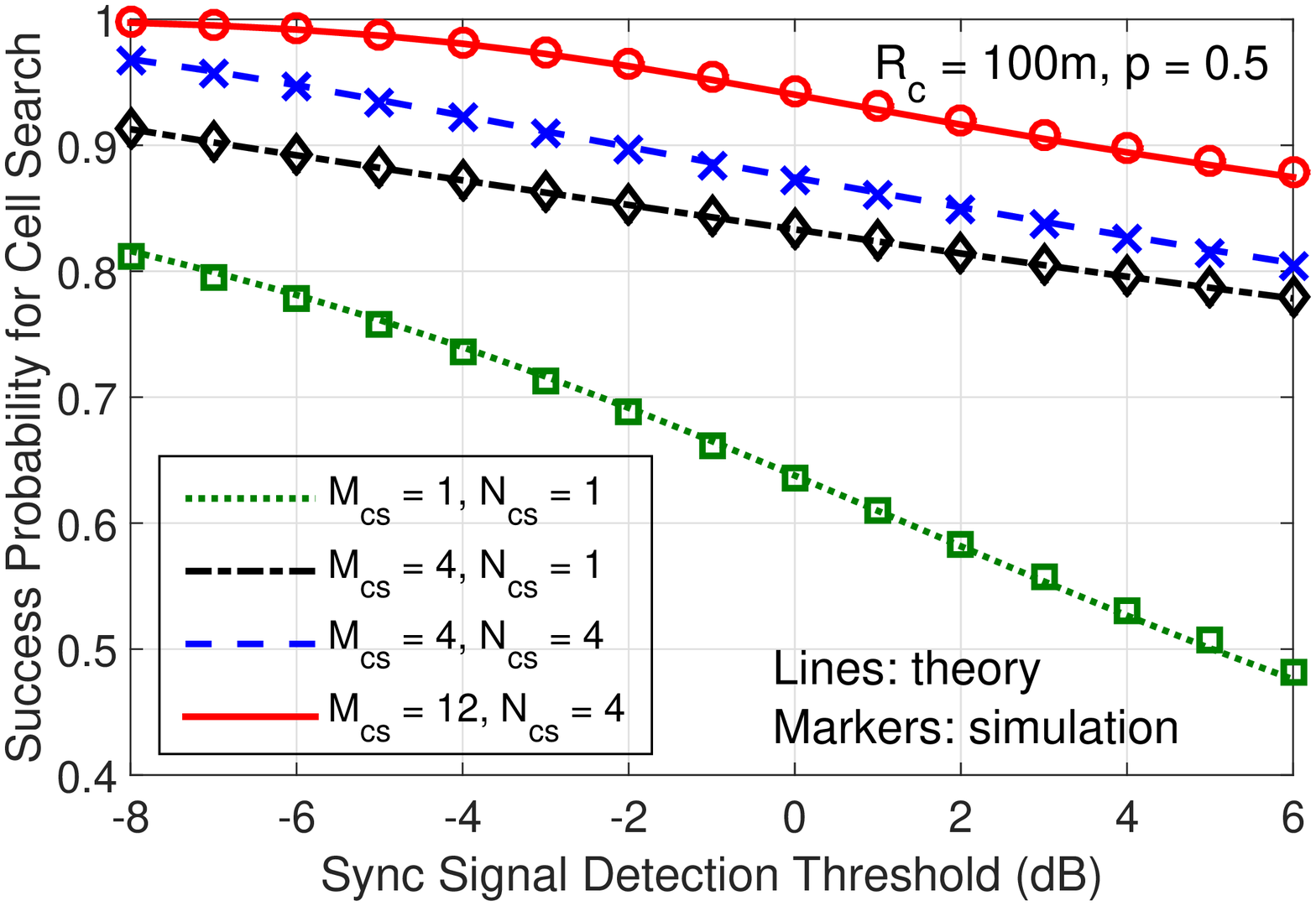}
     	\caption{Generalized LOS ball model }
		\label{SP_CS_LOS_Fig}
	\end{subfigure}
	\hfill
	\begin{subfigure}[b]{0.46\textwidth}
		\centering
		\includegraphics[height=1.7in, width=2.6in]{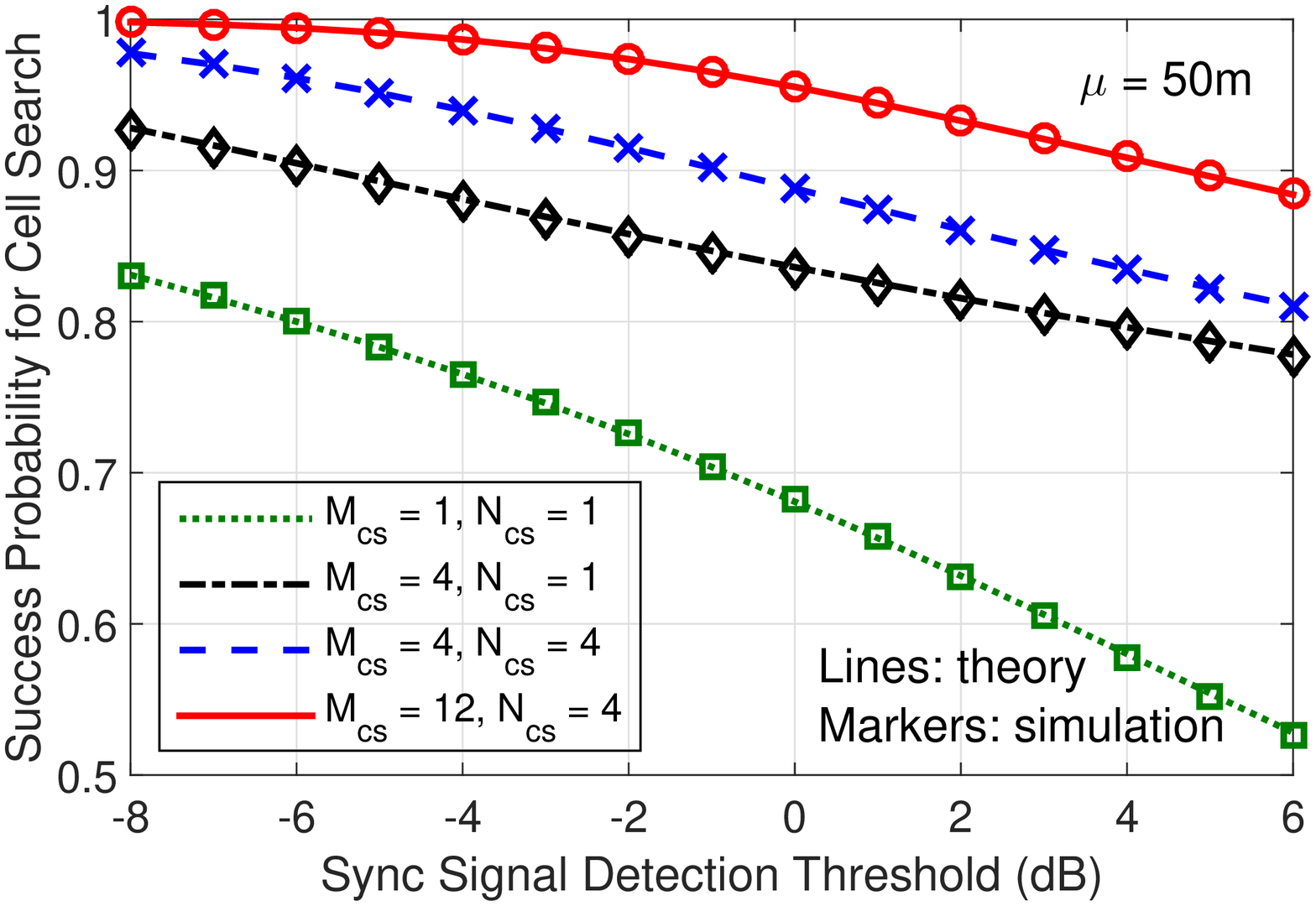}
		\caption{Exponential blockage model }
		\label{SP_CS_EXP_Fig}
	\end{subfigure}
	\caption{Cell search success probability.}\label{SP_CS_Fig}
\end{figure}

Fig.~\ref{ServingPL_CCDF_Fig} shows the CCDF of the path loss from the typical user to the tagged BS, which is derived in Lemma~\ref{minPLdistlemma}. As we increase $M_{cs}$ or $N_{cs}$, the CCDF of the path loss decreases, especially at the tail of the distribution. Note when $M_{cs} = N_{cs} = 1$, the tagged BS is the BS providing the minimum path loss to the typical user, which coincides with the conventional minimum path loss association rule~\cite{trac,akoum2012covrage,bai2015coverage,singh2015tractable,alkhateeb2016initial}. By contrast, if beam sweeping is implemented for cell search, the typical user can connect to other BSs even if the BS providing the minimum path loss is unable to be detected. As a result, the typical user will have a smaller path loss to the tagged BS almost surely as $M_{cs}$ or $N_{cs}$ increases. This fact further demonstrates the benefit of beam sweeping for cell search. Actually, all the CCDF curves in Fig.~\ref{ServingPL_CCDF_LOS_Fig} have an inflection point at $101.4$ dB. This is because for the LOS ball blockage model, the serving BS could be either LOS or NLOS when the path loss is smaller than $101.4$ dB, while it is NLOS almost surely when the path loss is higher than $101.4$ dB.

\begin{figure}[h]
	\begin{subfigure}[b]{0.46\textwidth}
		\centering
		\includegraphics[height=1.7in, width= 2.6in]{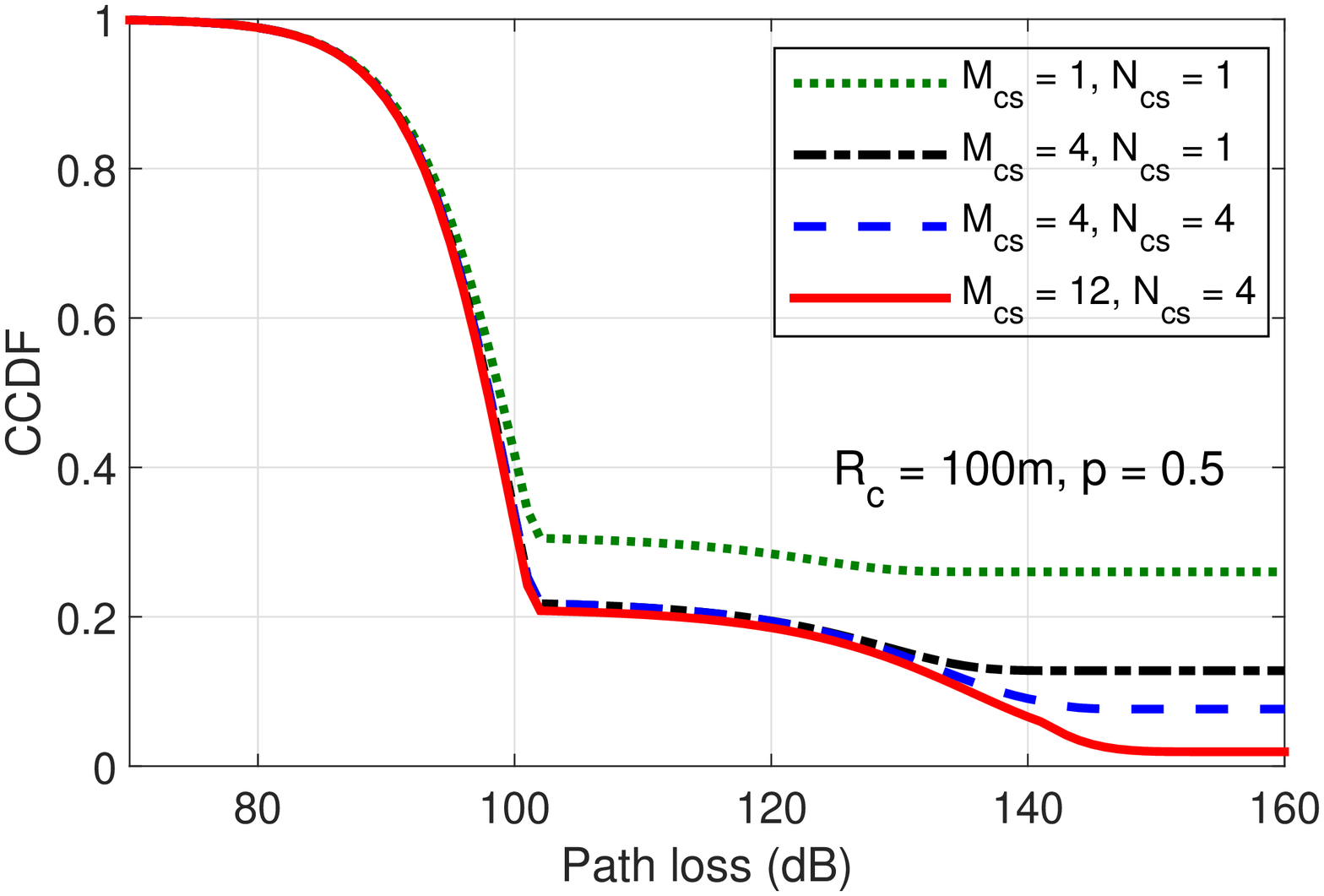}
	    \caption{Generalized LOS ball model}
		\label{ServingPL_CCDF_LOS_Fig}
	\end{subfigure}
	\hfill
	\begin{subfigure}[b]{0.46\textwidth}
		\centering
		\includegraphics[height=1.7in, width=2.6in]{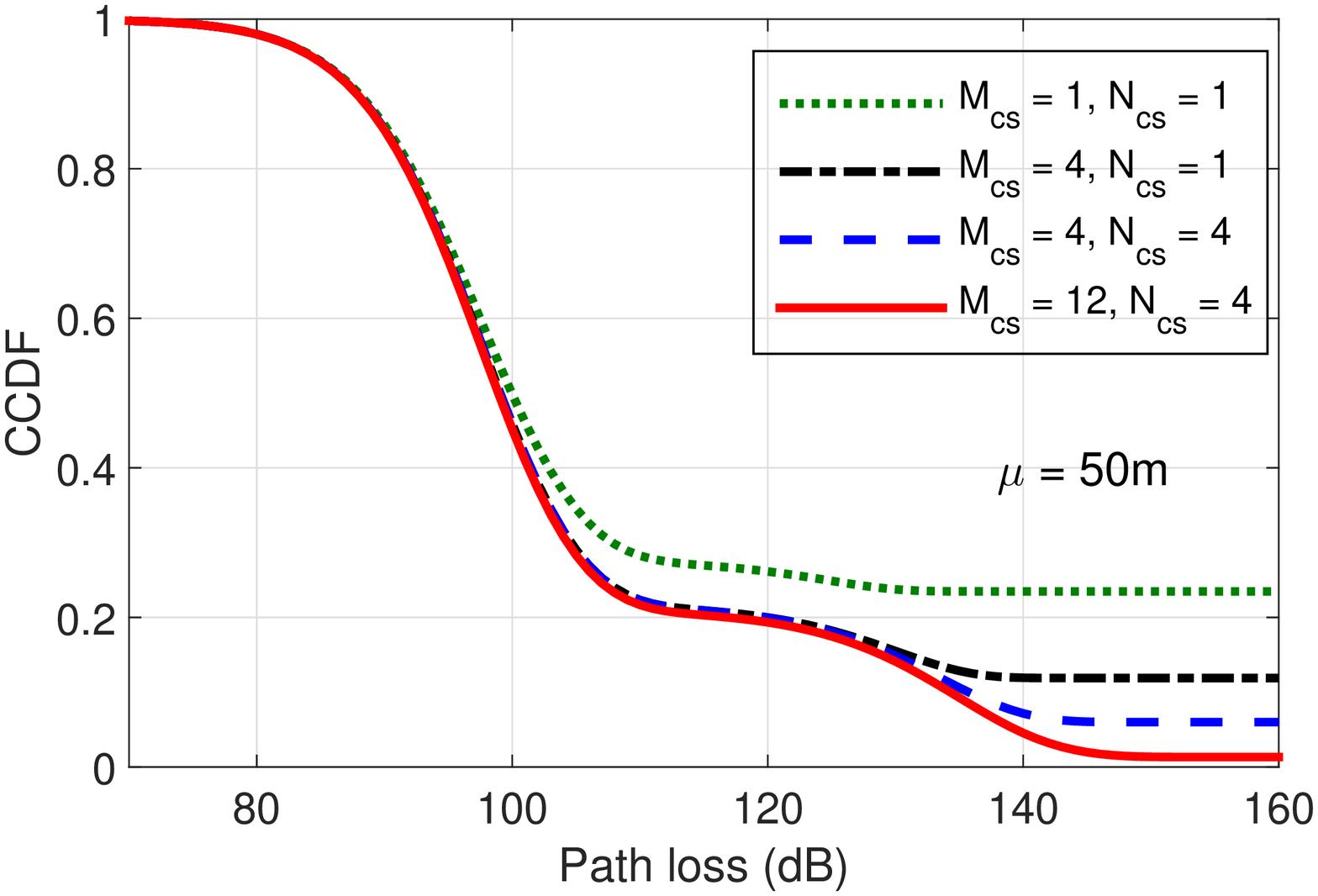}
		\caption{Exponential blockage model}
	    \label{ServingPL_CCDF_EXP_Fig}
	\end{subfigure}
	\caption{Path loss distribution from the typical user to tagged BS. }\label{ServingPL_CCDF_Fig}
\end{figure}

\subsubsection{No RA Preamble Collision Probability} Fig.~\ref{No_PA_Collision_Prob_Fig} plots the probability that the typical user is not subject to RA preamble collisions versus the number of BS beams $M = \max(M_{cs},M_{ra})$. Different parameters for the two blockage models are considered, where blockage becomes more severe as $p$ decreases in the generalized LOS ball model, or $\mu$ decreases in the exponential blockage model. It can be observed from Fig.~\ref{No_PA_Collision_Prob_Fig}
that Lemma~\ref{No_PA_Collision_Lemma} is an accurate approximation to the actual simulation results, which shows that it is accurate to approximate the users that succeed cell search by PPP. In addition, Fig.~\ref{No_PA_Collision_Prob_Fig} shows that the probability of no RA preamble collision $P_{co}$ is relatively insensitive to the underlying blockage conditions, and $P_{co}$ increases as the number of BS beams increases. Since $P_{co}$ remains consistently high (greater than 95\%) for different blockage conditions and various $M$ values, RA preamble collision is therefore not the performance bottleneck for the baseline protocol. It is clear from Lemma~\ref{No_PA_Collision_Lemma} that this is a result of the 64 RA preamble sequences and beam sweeping at the BS during random access.

\begin{figure}
	\begin{subfigure}[b]{0.46\textwidth}
		\centering
		\includegraphics[height=1.85in, width= 3.0in]{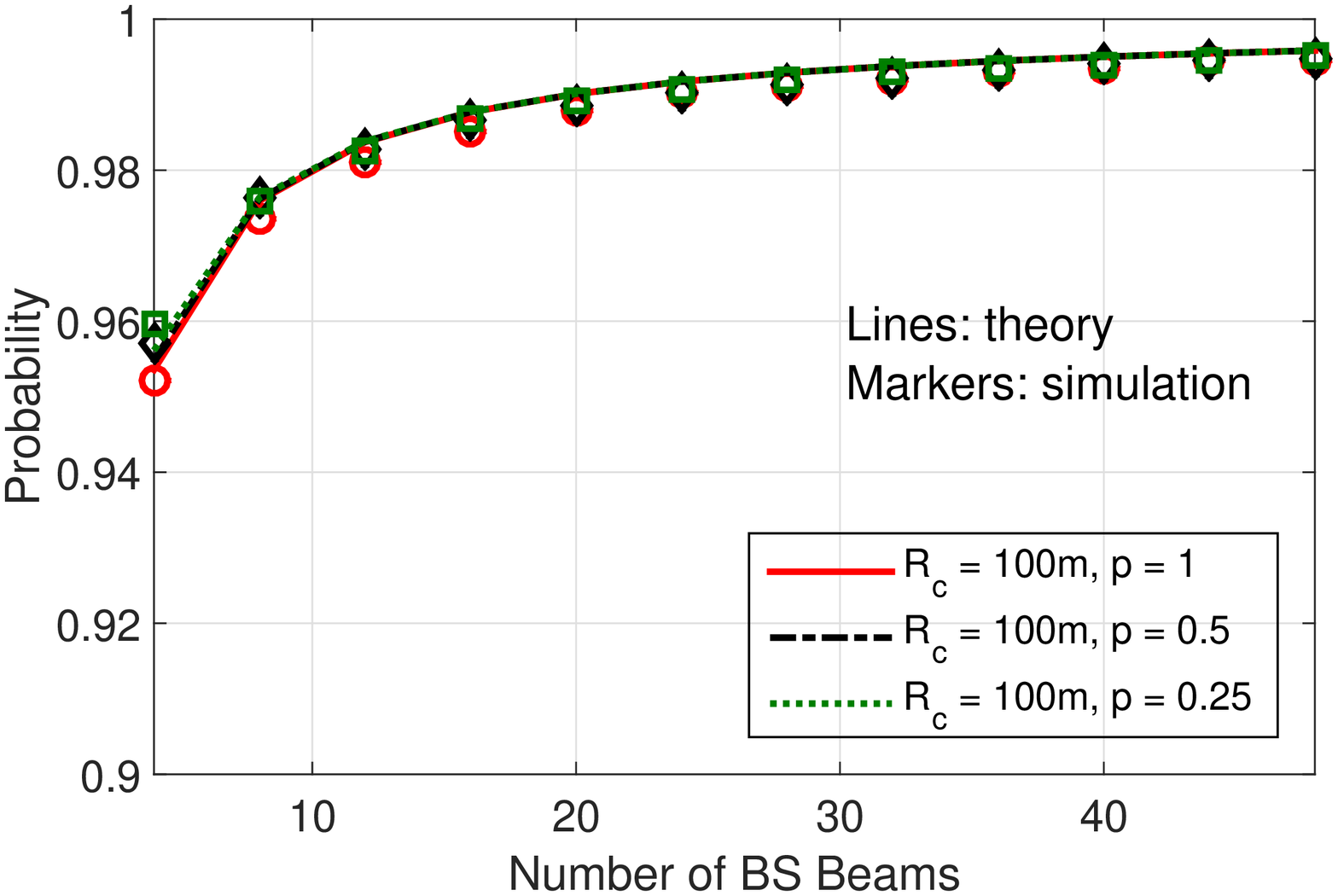}
		\caption{Generalized LOS ball model}
		\label{No_PA_Collision_Prob_LOS_Fig}
	\end{subfigure}
	\hfill
	\begin{subfigure}[b]{0.48\textwidth}
		\centering
		\includegraphics[height=1.7in, width= 2.6in]{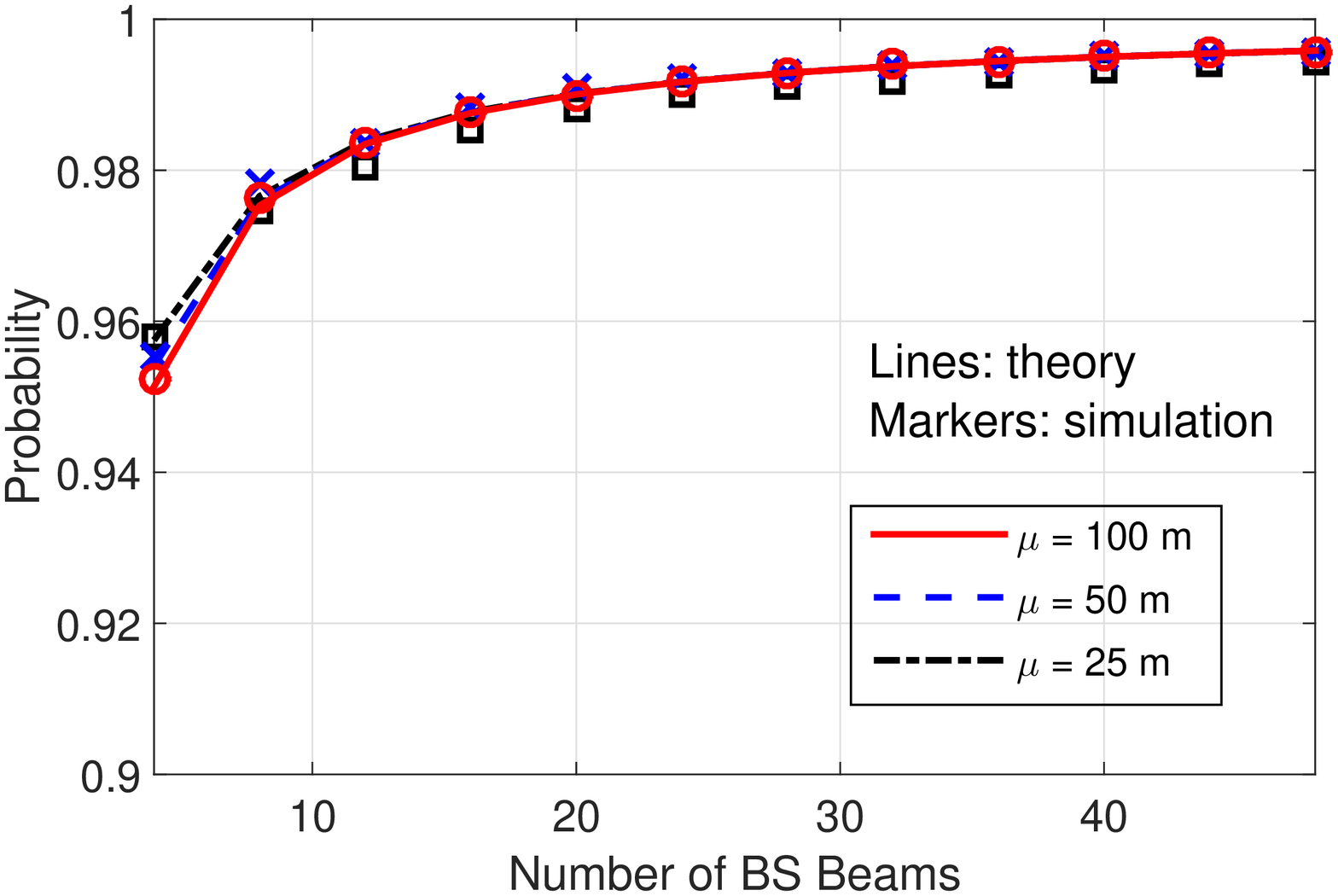}
		\caption{Exponential blockage model}
		\label{No_PA_Collision_Prob_EXP_Fig}
	\end{subfigure}
	\caption{Probability of no RA preamble collision.}\label{No_PA_Collision_Prob_Fig}
\end{figure}

\subsubsection{Expected Initial Access Delay} The initial access delay, which can be derived from Theorem~\ref{Overall_IA_SP_Thm}, is plotted in Fig.~\ref{Exp_IA_Delay_Fig} for both blockage models. Despite some approximations used in deriving Theorem~\ref{Overall_IA_SP_Thm}, Fig.~\ref{Exp_IA_Delay_Fig} validates the  accuracy of the analytical results. 
In addition, Fig.~\ref{Exp_IA_Delay_Fig} shows that as blockage becomes less severe, the expected initial access delay decreases as a result of the improved initial access success probability. Depending on the propagation environment, the optimal expected initial access delay in our simulations ranges from 2.2 ms to 4.1 ms for the generalized LOS ball model, and 1.2 ms to 5.0 ms for the exponential blockage model. 

According to Fig.~\ref{Exp_IA_Delay_Fig}, the expected initial access delay is relatively high when the number of BS beams is small, which is because the typical user needs more initial access cycles until it can connect to the network. By increasing the number of BS beams, despite the typical user has higher probability to succeed within one initial access cycle, the overhead for initial access starts to become more dominant. As a result, there exists an optimal BS beam number (or BS beamwidth) in terms of the expected initial access delay. For example, given $R_c = 100$ m for the generalized LOS ball model, this optimal beamwidth is $45^{\circ}$, $22.5^{\circ}$ and $15^{\circ}$ when $p$ is equal to $1$, $0.5$ and $0.25$ respectively. In fact, as blockage becomes more severe, the optimal BS beamwidth is decreasing for both blockage models, which means a more robust link with higher antenna gain is needed in order to quickly establish the connection. In addition, we have also verified that the optimal BS beamwidth in terms of initial access delay is non-decreasing as user density increases, which is mainly because a narrower beam at the BS will reduce the collision of RA preambles among different users.  

\begin{figure}
	\begin{subfigure}[b]{0.46\textwidth}
		\centering
		\includegraphics[height=1.7in, width= 2.8in]{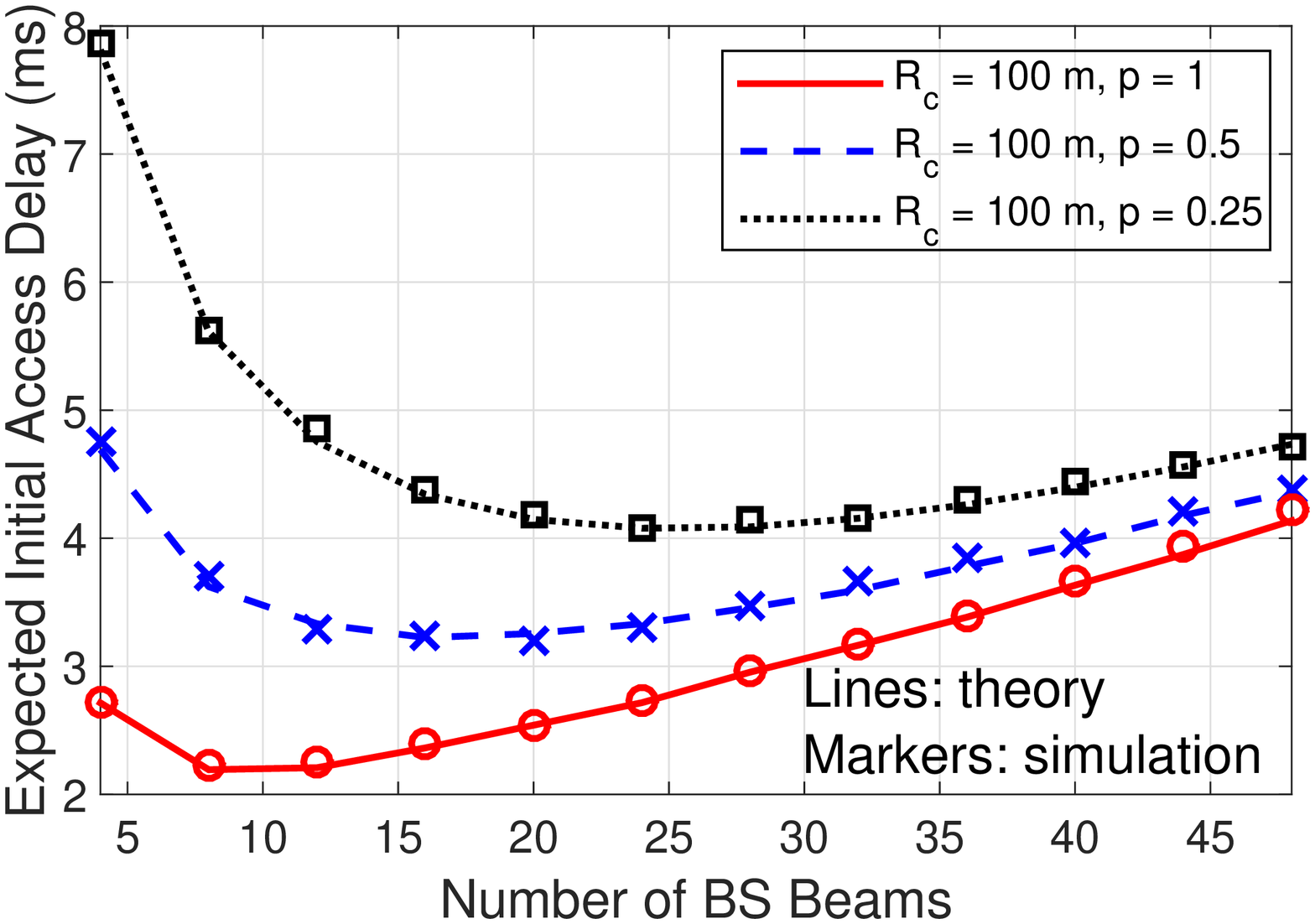}
		\caption{Generalized LOS ball model}
		\label{Exp_IA_Delay_LOS_Fig}
	\end{subfigure}
	\hfill
	\begin{subfigure}[b]{0.46\textwidth}
		\centering
		\includegraphics[height=1.7in, width=2.8in]{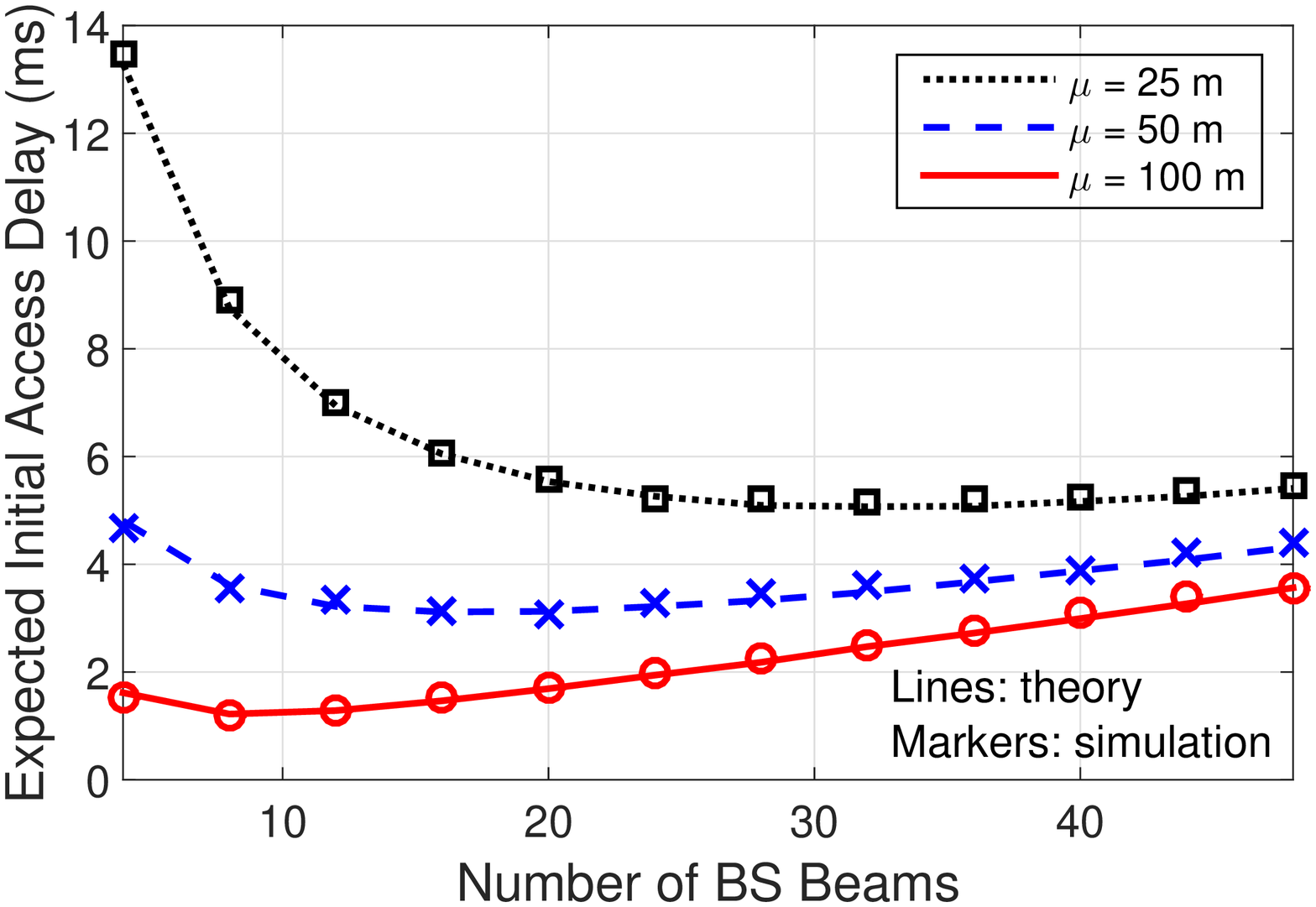}
		\caption{Exponential blockage model }
		\label{Exp_IA_Delay_EXP_Fig}
	\end{subfigure}
	\caption{Expected initial access delay performance.}\label{Exp_IA_Delay_Fig}
\end{figure}

\subsection{Performance for the Data Transmission Phase} Fig.~\ref{DataSINRCCDF_Fig} plots the downlink data SINR coverage probability given the typical user succeeds the initial access, where the BS and user beamwidth are $30^\circ$ and $90^\circ$ respectively. Although all interfering BSs are assumed to be active in deriving Lemma~\ref{CCDF_Data_SINR_Lemma}, the difference between the analytical and simulation results in Fig.~\ref{DataSINRCCDF_Fig} is negligible in the range of parameters chosen here. The same trend has been observed for other BS and user beamwidth values as well, which validates the accuracy of Lemma~\ref{CCDF_Data_SINR_Lemma}. 
In addition, we can observe non-concave behavior of the data SINR curves in Fig.~\ref{DataSINRCCDF_Fig}, which is because the overall CCDF of data SINR in Fig.~\ref{DataSINRCCDF_Fig} is obtained by adding up the CCDF of data SINR when the serving BS is LOS and NLOS respectively. 

The average user-perceived downlink throughput versus the number of BS beams is plotted in Fig.~\ref{UPT_Fig} for both blockage models. Fig.~\ref{UPT_Fig} shows the average UPT has a steep increase when the number of BS beams increases from a very small value to a medium value. This is mainly due to a much improved link quality and relatively low initial access overhead in this range. However, as the number of BS beams further increases, the initial access overhead starts to become more dominant, which leads to a steady decrease of the average UPT. In terms of the average UPT, Fig.~\ref{UPT_Fig} shows the optimal BS beamwidth does not vary too much for different blockage conditions, which is typically between 10$^\circ$ to 18$^\circ$. 
This is because the average UPT is affected by multiple counterbalancing factors such as the initial access overhead, success probability of initial access, and scheduling factors. For example, a high initial access success probability will lead to a heavily-loaded cell for the tagged BS, such that the typical user has smaller probability to be scheduled.

Therefore, for the baseline initial access protocol, depending on the blockage condition and which metric is more important, the optimal BS beamwidth could vary. When blockage is not very significant, a wide BS beamwidth (e.g. 45$^\circ$) is preferred to reduce the initial access delay, while a narrow BS beamwidth (e.g. 15$^\circ$) is preferred to achieve higher UPT performance. By contrast, when blockage is severe, a narrow BS beamwidth (e.g. 15$^\circ$) could achieve good performance for both initial access delay and average UPT. 

\begin{figure}
	\begin{subfigure}[b]{0.46\textwidth}
		\centering
		\includegraphics[height=1.7in, width= 2.8in]{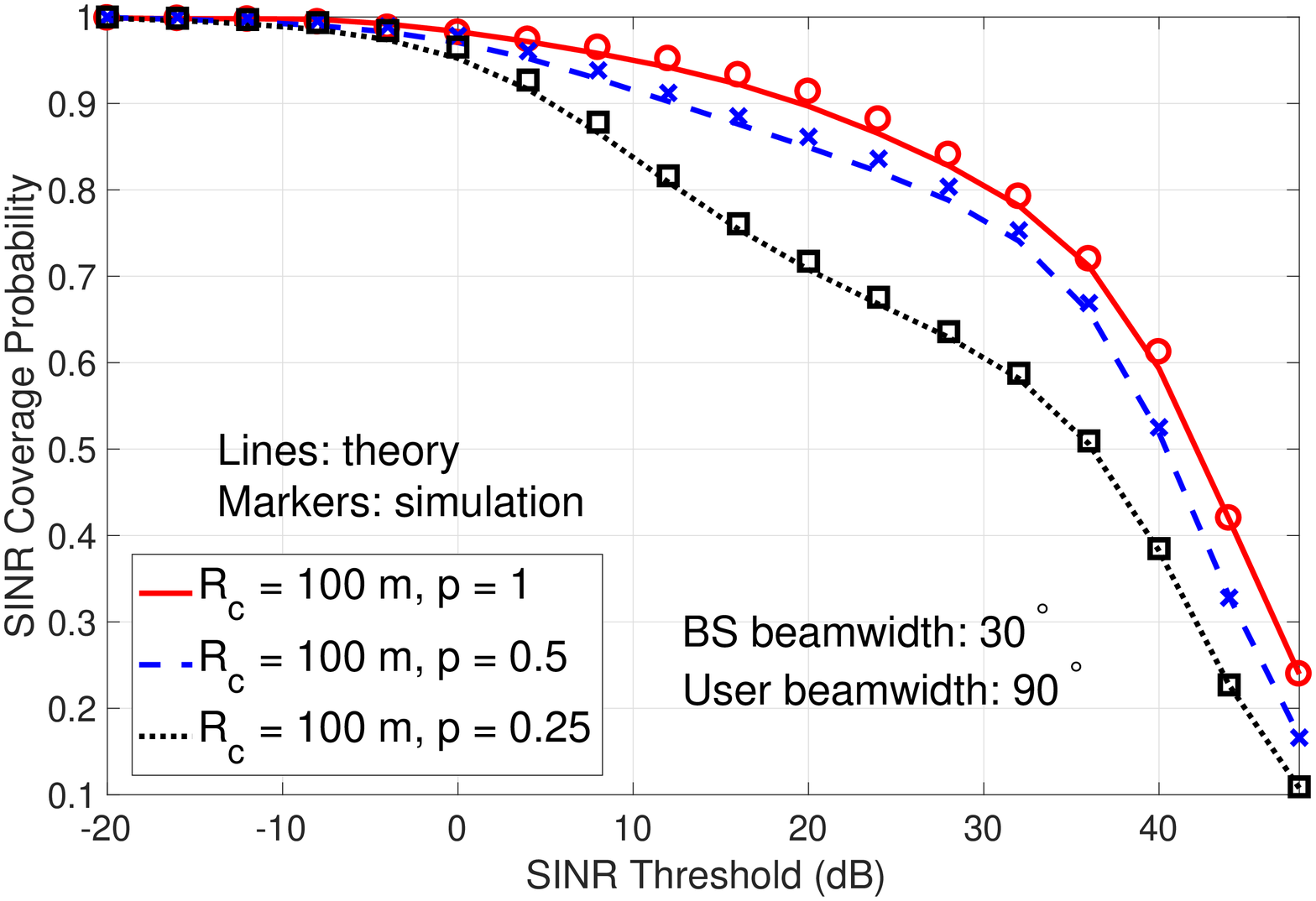}
		\caption{Generalized LOS ball model}
		\label{DataSINRCCDF_LOS_Fig}
	\end{subfigure}
	\hfill
	\begin{subfigure}[b]{0.46\textwidth}
		\centering
		\includegraphics[height=1.7in, width=2.8in]{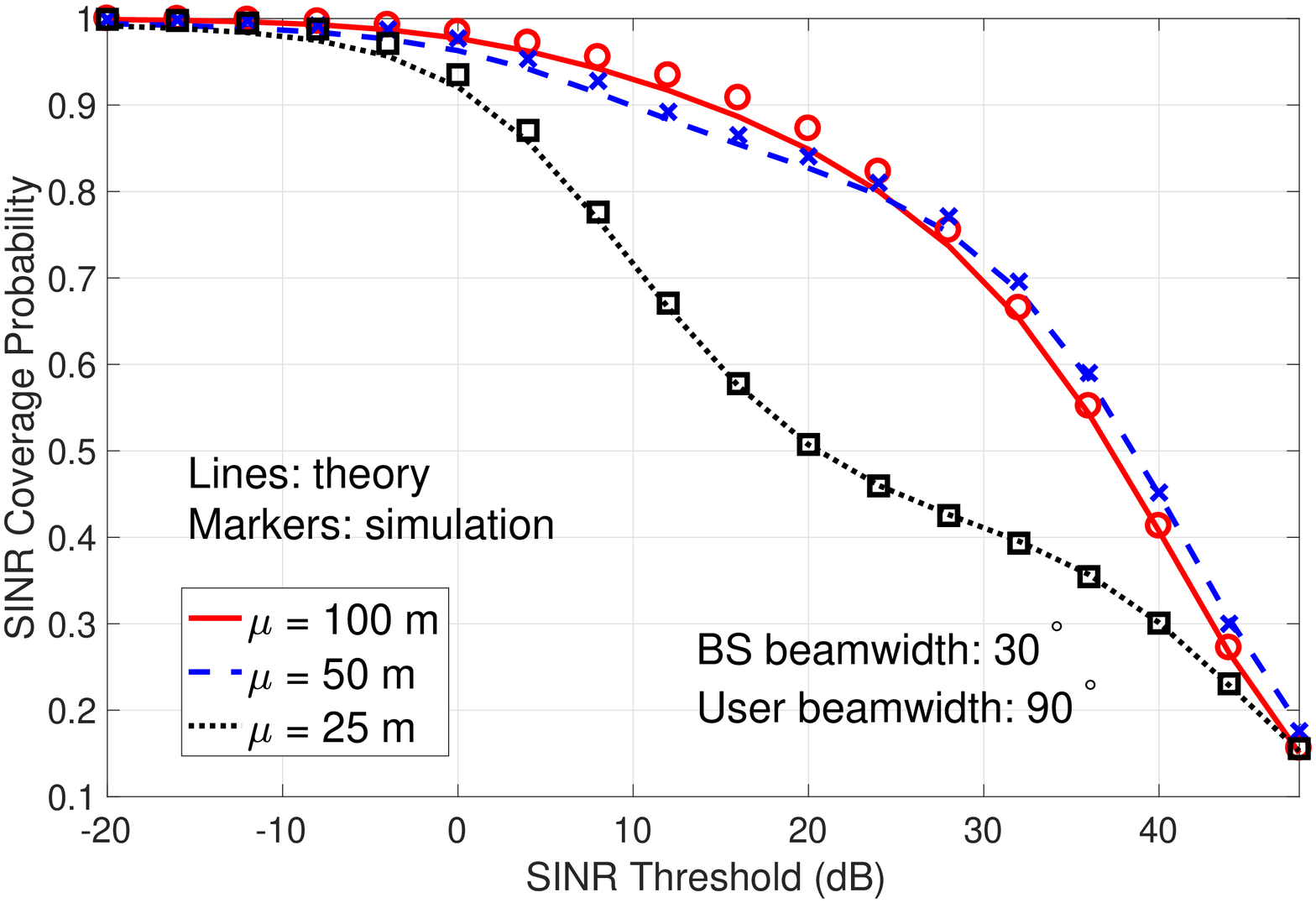}
		\caption{Exponential blockage model}
		\label{DataSINRCCDF_EXP_Fig}
	\end{subfigure}
	\caption{CCDF of data SINR given successful initial access.}\label{DataSINRCCDF_Fig}
\end{figure}

\begin{figure}
	\begin{subfigure}[b]{0.46\textwidth}
		\centering
		\includegraphics[height=1.7in, width= 2.8in]{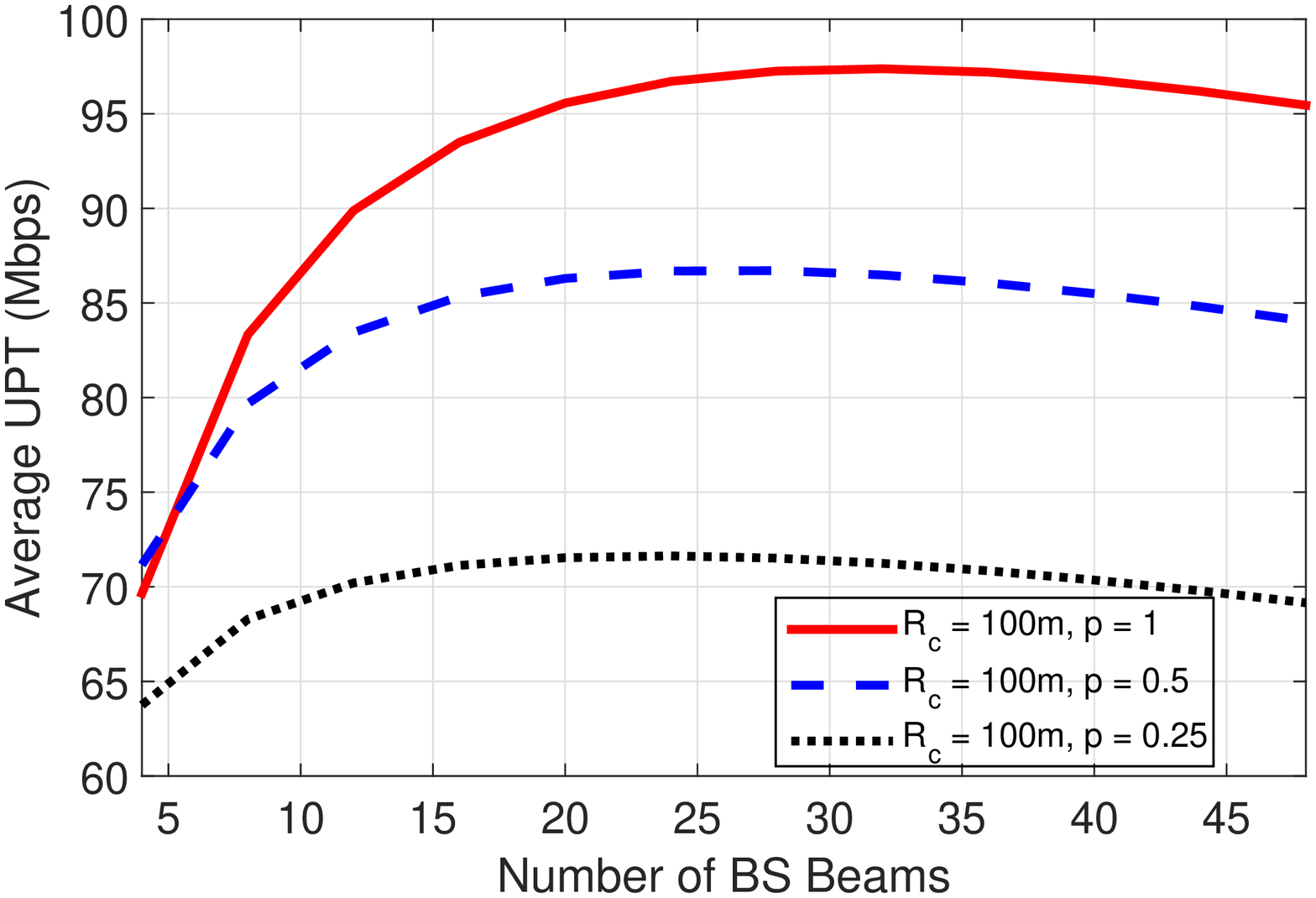}
		\caption{Generalized LOS ball model}
		\label{UPT_LOS_Fig}
	\end{subfigure}
	\hfill
	\begin{subfigure}[b]{0.46\textwidth}
		\centering
		\includegraphics[height=1.7in, width=2.8in]{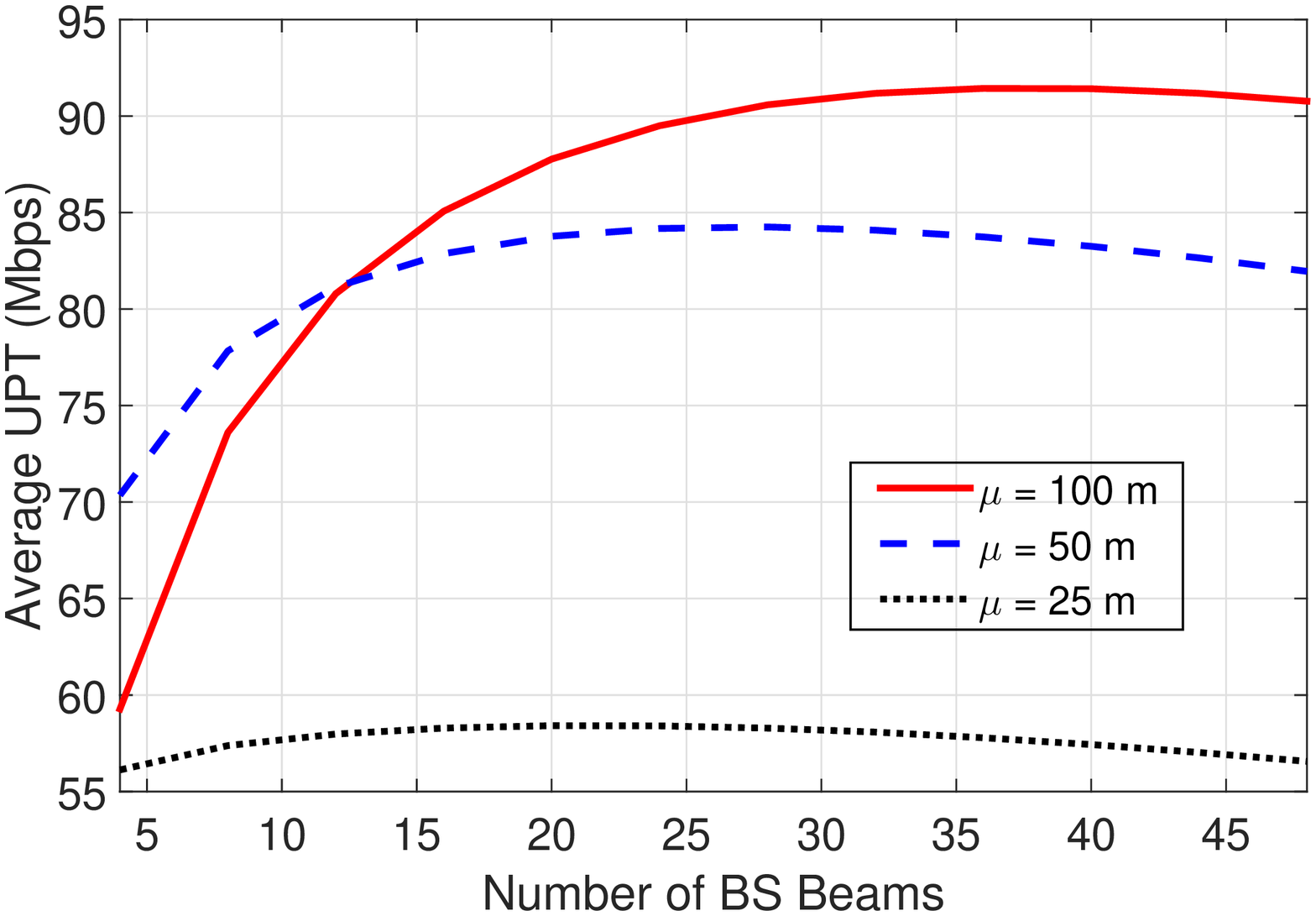}
		\caption{Exponential blockage model}
		\label{UPT_EXP_Fig}
	\end{subfigure}
	\caption{Average user-perceived downlink throughput performance.}\label{UPT_Fig}
\end{figure}

\section{Performance Comparison for Different Initial Access Protocols}\label{IAOptPerfCompSec} 
In this section, based on Theorem~\ref{Overall_IA_SP_Thm} and Theorem~\ref{Avg_UPT_thm}, we compare the expected initial access delay and average user-perceived downlink throughput\footnote{The accuracy of the analytical results for the other three protocols can be validated similar to the baseline protocol, so only analytical results are shown in this section.} for the four initial access protocols in Table~\ref{IAOptions}. 
In making the comparisons, we consider both a severely blocked condition (e.g., $R_c = 100$ m, $p = 0.25$ for generalized LOS ball model; $\mu = 25$ m for exponential blockage model), and also a lightly blocked condition (e.g., $R_c = 100$ m, $p = 1$ for generalized LOS ball model; $\mu = 100$ m for exponential blockage model) for both blockage models. All the other system parameters remain the same as in Table~\ref{SysParaTable}. 

\subsection{Comparison of Expected Initial Access Delay} 
\textbf{Baseline and fast CS outperform fast RA and omni RX in terms of the expected initial access delay.} The expected initial access delay is plotted in Fig.~\ref{IAdelay_comp_fig}, which shows that the fast RA and omni RX protocols always have higher initial access delay than the baseline and fast CS protocols. In addition, as the BS beam number $M$ increases, the expected IA delay for the fast RA and omni RX protocols will keep increasing, while the expected IA delay for baseline and fast CS first decreases then increases. The main reason is that both fast RA and omni RX protocols require the BS to receive omni-directionally during random access (i.e., $M_{ra}=1$), which leads to a high RA preamble collision probability according to Lemma~\ref{No_PA_Collision_Lemma}. Specifically, in contrast to the baseline and fast CS protocols wherein $P_{co}$ is consistently higher than 0.95, $P_{co}$ for the fast RA and omni RX protocols is around 0.82 for all blockage conditions and various BS beamwidth values. In addition, the RA preamble decoding probability for the fast RA and omni RX  protocols is also lower (around 5\%) compared to the other two protocols, which can be analytically shown from Lemma~\ref{RA_SINR_COP_Lemma}. As a result, as we increase $M$, the overall initial access success probability is not significantly improved for fast RA and omni RX, but the initial access overhead increases. According to~(\ref{ExpDelayExpr}), the initial access delay for fast RA and omni RX will keep increasing as $M$ increases, which is significantly higher than the other two protocols. 

\textbf{The baseline protocol has smaller expected initial access delay than the fast CS protocol when blockage is severe.} In terms of the expected initial access delay, Fig.~\ref{IAdelay_comp_fig} also shows whether or not the baseline protocol outperforms the fast CS protocol depends on the severity of blockage. Specifically, under a severely blocked condition, the baseline protocol has smaller initial access delay than the fast CS protocol. This is mainly because the baseline protocol has better link quality under both cell search and random access phases, which leads to significantly higher initial access success probability than the fast CS protocol. By contrast, under light blockage, the fast CS protocol is also able to achieve a sufficiently high initial access success probability. As a result, the fast CS protocol will outperform the baseline protocol due to a much lower initial access overhead. 
 
\begin{figure}
	\begin{subfigure}[b]{0.46\textwidth}
		\centering
		\includegraphics[height=1.7in, width= 2.8in]{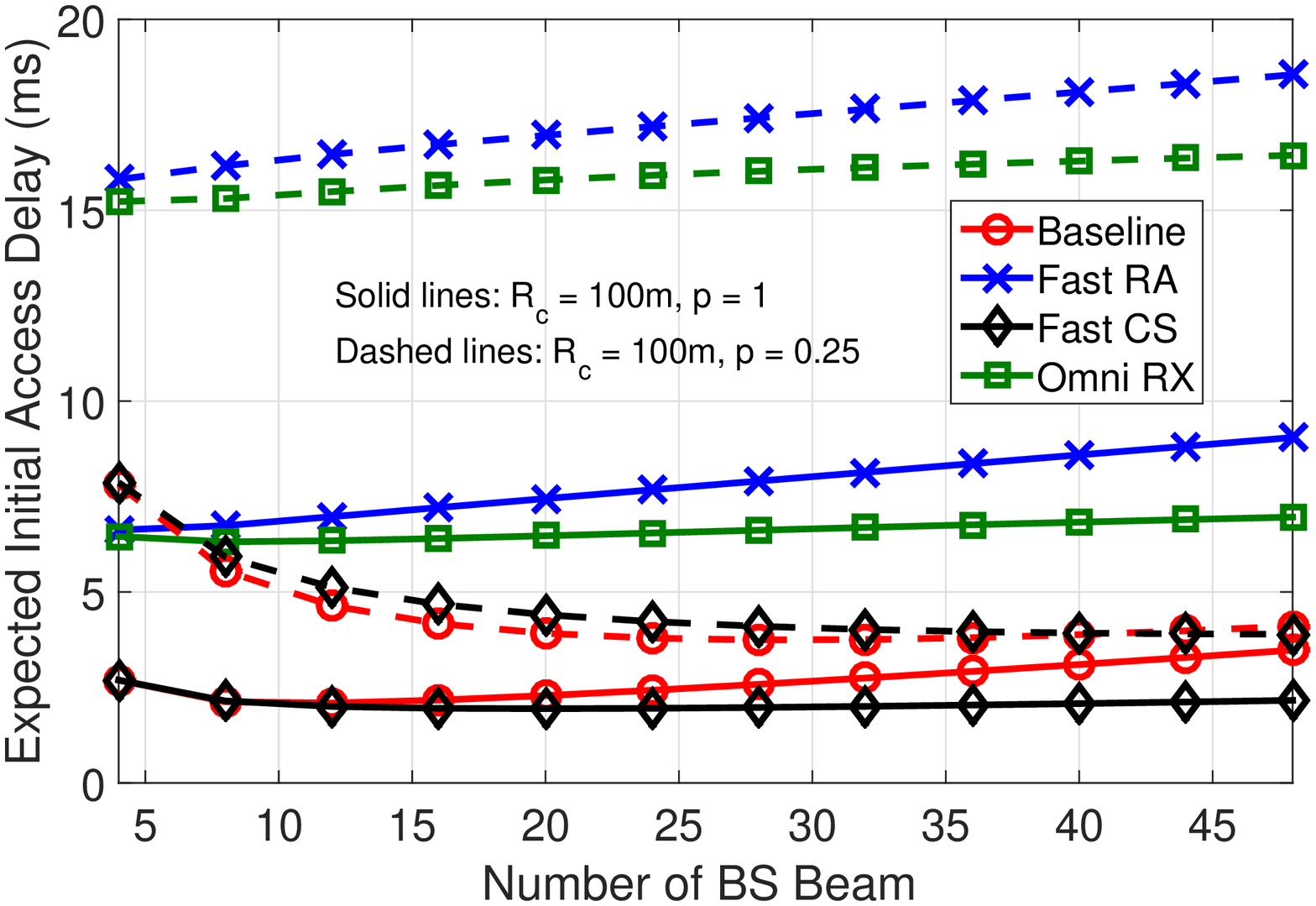}
		\caption{Generalized LOS ball model }   
		\label{IAdelay_comp_LOS_fig}
	\end{subfigure}
	\hfill
	\begin{subfigure}[b]{0.46\textwidth}
		\centering
		\includegraphics[height=1.7in, width=2.8in]{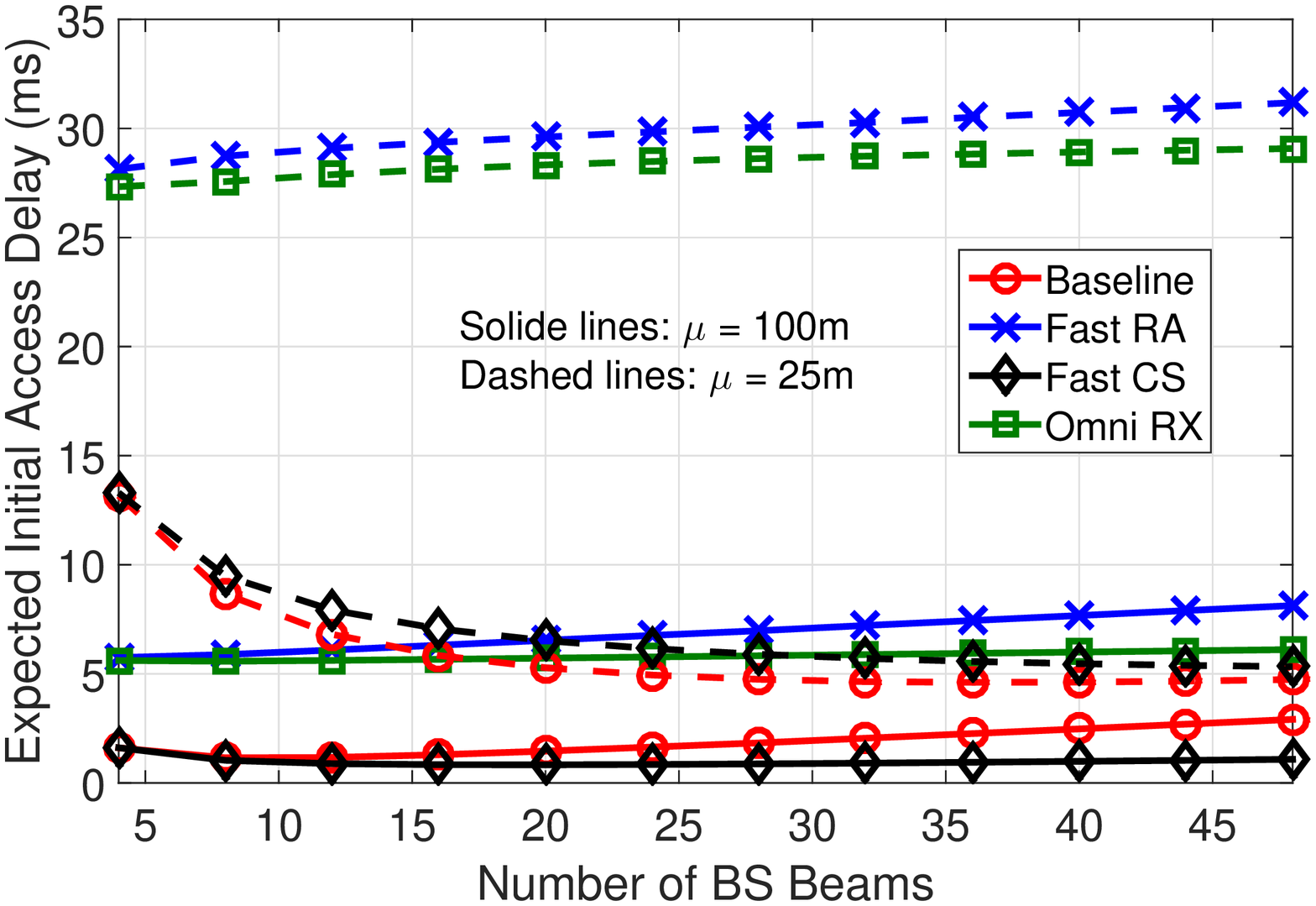}
		\caption{Exponential blockage model}
		\label{IAdelay_comp_EXP_fig}
	\end{subfigure}
	\caption{Comparison of expected initial access delay.}\label{IAdelay_comp_fig}
\end{figure}

\subsection{Comparison of the Average User-perceived Downlink Throughput} \textbf{In terms of the average UPT, the omni RX protocol and fast CS protocol generally outperform the other two protocols, while the baseline protocol only provides the smallest average UPT.} 
Despite having smaller initial access success probability than the baseline, the following reasons contribute to the high UPT for omni RX and fast CS protocols. First, both protocols have low initial access overhead, which is significantly lower than the other two protocols, especially when the number of BS beams $M$ is high. Second, despite both protocols having lower initial access success probability than the baseline, the typical user actually has higher scheduling probability since it will observe a lightly loaded cell once it succeeds at initial access. Third, under the omni RX and fast CS protocols, the typical user will have higher conditional downlink data SINR than the baseline, because it needs more favorable propagation in order to succeed at initial access. Despite the baseline protocol having the highest initial access success probability, the above factors render it inferior in terms of UPT versus the other protocols. 
Compares to the omni RX protocol, Fig.~\ref{UPT_Comp_Fig} shows that the fast CS protocol achieves a slightly higher UPT performance under a lightly blocked condition, while it provides a much smaller UPT under a severely blocked condition. 

Another observation from Fig.~\ref{UPT_Comp_Fig} is that as the BS beam number $M$ increases, the UPT for baseline and fast RA protocols first increases as a result of better data SINR and relatively low initial access overhead, then UPT will start to decrease as the initial access overhead term becomes more significant. In fact, the same trend also applies to omni RX and fast CS protocols as we continue to increase the number of BS beams in Fig.~\ref{UPT_Comp_Fig}, but the corresponding optimal beamwidth could be too narrow to implement in a real system. 

In summary, the baseline protocol is mainly beneficial for delay-sensitive applications since it provides a small initial access delay, especially when blockage is severe. However, due to the high initial access overhead, the baseline protocol also has a poor user-perceived downlink throughput performance. The omni RX protocol provides the best user-perceived downlink throughput performance, but it is unlikely to be adopted unless the network is delay-tolerant. By contrast, the fast CS protocol generally gives a good trade-off between the initial access delay and user-perceived throughput performance, especially under a lightly blocked condition.

\begin{figure}
	\begin{subfigure}[b]{0.46\textwidth}
		\centering
		\includegraphics[height=1.7in, width= 2.8in]{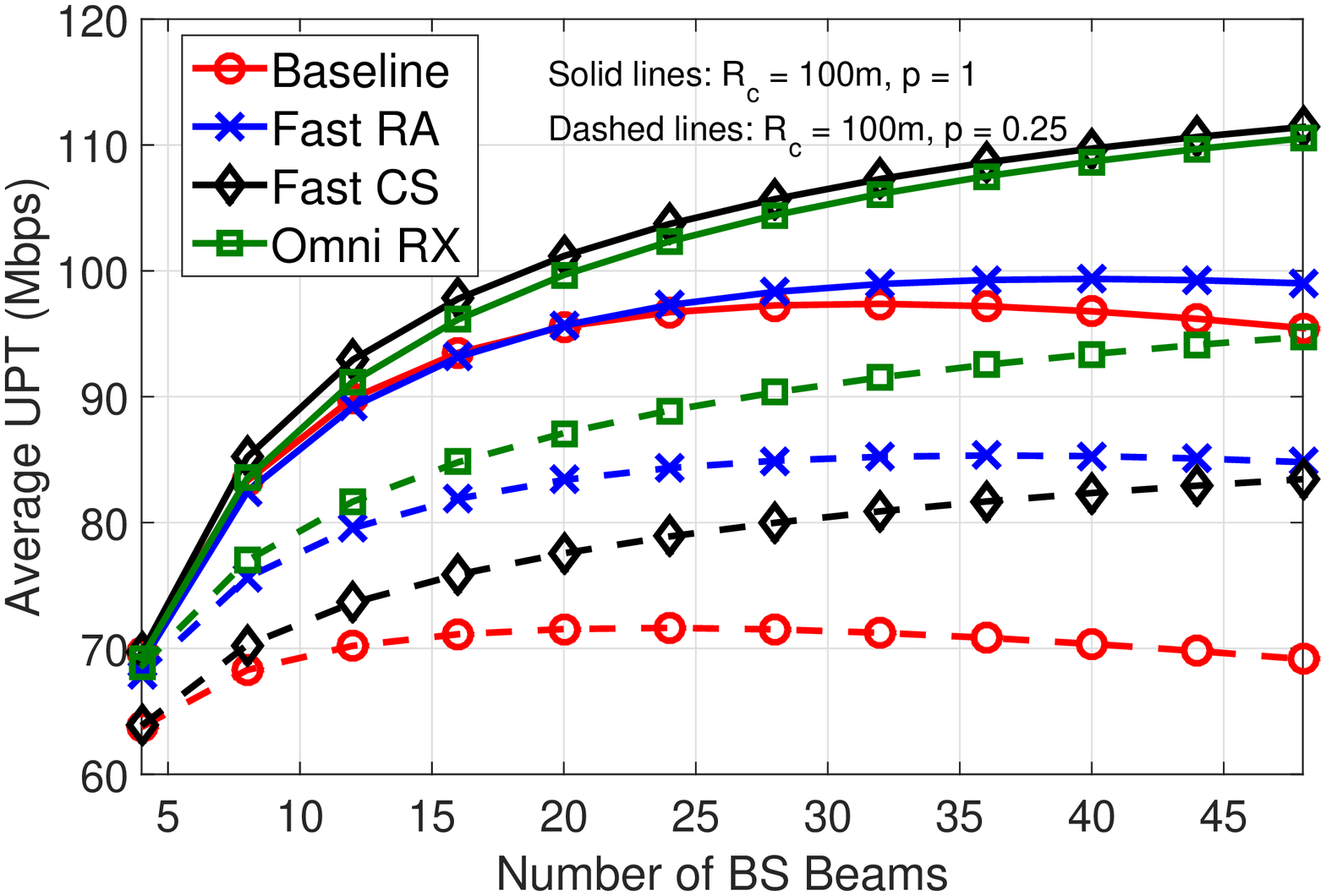}
		\caption{Generalized LOS ball model }   
		\label{UPT_Comp_LOS_Fig}
	\end{subfigure}
	\hfill
	\begin{subfigure}[b]{0.46\textwidth}
		\centering
		\includegraphics[height=1.7in, width=2.8in]{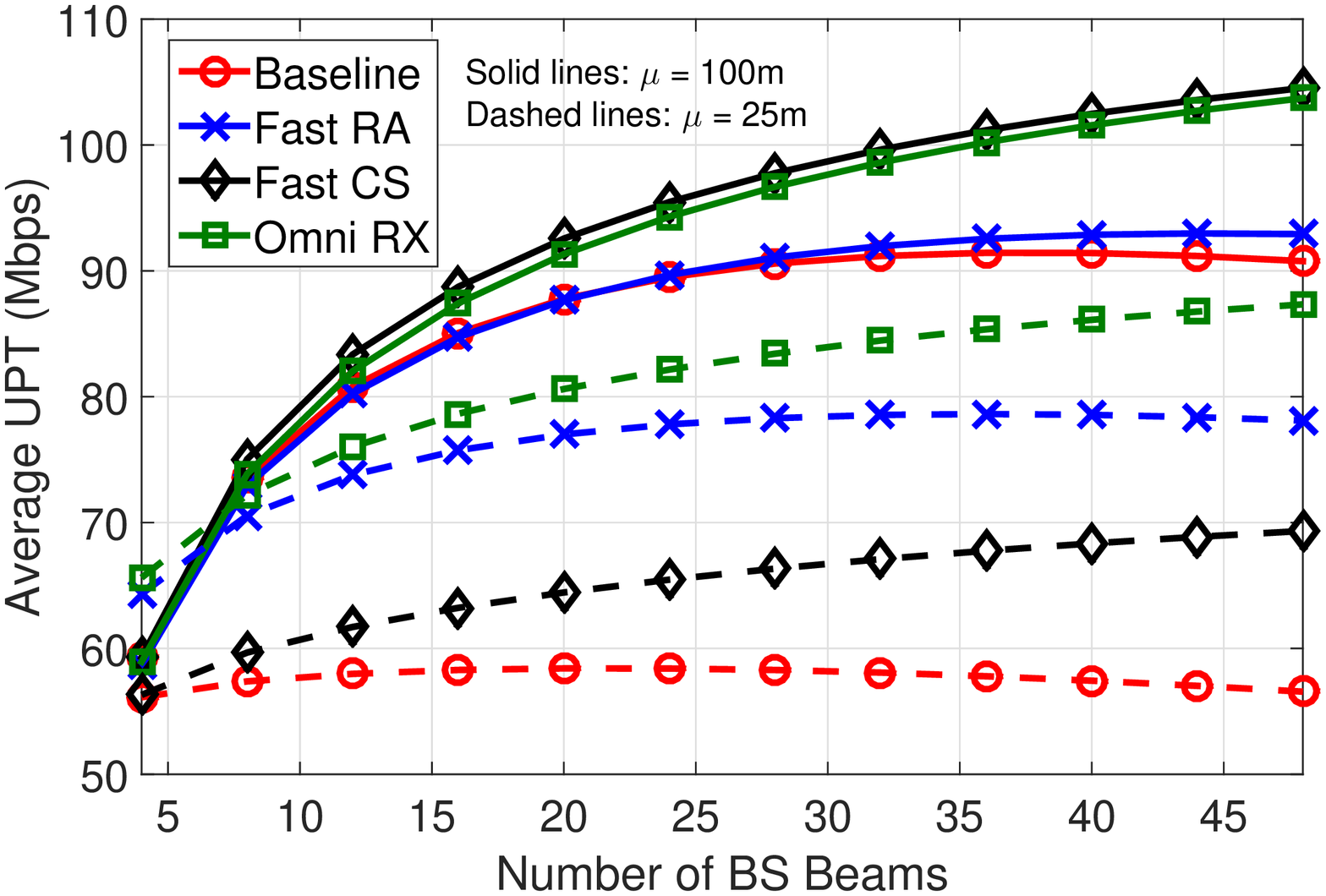}
		\caption{Exponential blockage model}
		\label{UPT_Comp_EXP_Fig}
	\end{subfigure}
	\caption{Comparison of average user-perceived downlink throughput.}\label{UPT_Comp_Fig}
\end{figure}

\section{Conclusions and Discussions}
This paper has proposed an analytical framework to investigate the effects of initial access protocol design on the system level performance of mmWave cellular networks. In particular, several new system design insights have been obtained by deriving the expected initial access delay, as well as a new metric called average user-perceived downlink throughput for four sample initial access protocols. Specifically, compared to LTE, where initial access is performed omni-directionally, we have shown that beam sweeping needs to be applied to cell search for mmWave cellular networks. For the baseline exhaustive search protocol, as blockage becomes more severe, the optimal BS beamwidth in terms of the initial access delay is found to decrease. However, the optimal BS beamwidth in terms of user-perceived downlink throughput is fairly constant and is within $[10^\circ, 18^\circ]$. Among the four initial access protocols that are investigated, the baseline protocol achieves the smallest initial access delay when blockage is severe, while the omni RX protocol generally provides the highest user-perceived downlink throughput. By contrast, the best trade-off between the initial access delay and user-perceived downlink throughput can be achieved by the fast CS protocol, wherein BSs transmit with relatively wide beams and users apply beam sweeping during cell search. 

By investigating the basic initial access protocols, this paper has shown that mmWave network performance depends heavily on the initial access design. Future work can leverage the proposed analytical framework to investigate the performance of more enhanced initial access protocols such as hierarchical beamforming~\cite{ieee2012IEEEad}, or to optimize various beam sweeping approaches as opposed to this paper where either omni or complete beam sweeping is used. In addition, the characterization of the uplink throughput is also an important future topic to investigate. 

\section*{Acknowledgments}
This work is supported in part by the National Science Foundation under Grant No. NSF-CCF-1218338 and an award from the Simons Foundation (\#197982), both to the University of Texas at Austin.

\appendices
\section{Proof of Lemma~\ref{RA_SINR_COP_Lemma}} \label{RA_SINR_COP_Lemma_Proof_Apdx}
Without loss of generality, denote the location of the tagged BS by $x_0$, and assume the first transmit beam of the typical user and the $m$-th ($1 \leq m \leq M_{ra}$) receive beam of the tagged BS are aligned. Since RA preamble sequences are randomly chosen, the users with the same RA preamble as the typical user 
form a PPP with intensity $\lambda_u \frac{\hat{P}_{M_{cs},N_{cs}}(\Gamma_{cs})}{N_{PA}}$, which is denoted by $\Phi_u^{'}$. Depending on whether the link to the tagged BS with distance $r$ is LOS or not, the users in $\Phi_u^{'}$ are further divided into two non-homogeneous PPPs $\Phi_{u,L}^{'}$ and $\Phi_{u,N}^{'}$, with the intensities being $\lambda_u \frac{\hat{P}_{M_{cs},N_{cs}}(\Gamma_{cs})}{N_{PA}} h(r)$ and $\lambda_u \frac{\hat{P}_{M_{cs},N_{cs}}(\Gamma_{cs})}{N_{PA}}(1-h(r))$ respectively.

When the typical user and the tagged BS are beam aligned, the SINR of the typical user's RA preamble sequence at the tagged BS is given by:
\allowdisplaybreaks
\begin{align}\label{RA_SINR_COP_Lemma_Proof_Eq}
\allowdisplaybreaks
\!\!\!\! \text{SINR}_{PA}(Z_0) = \frac{F_0M_{ra} G(\frac{2\pi}{N})/Z_0}{\!\!\!\!\!\!\!\!\!\!\!\!\!\!\!\!\!\!\!\sum\limits_{u_{i}^L \in \Phi_{u,L}^{'} \cap S(x_0, \frac{2(m-1)\pi}{M_{ra}},\frac{2m\pi}{M_{ra}})}\!\!\!\!\!\!\!\!\!\!\!\!\!\!\!\!\!\! \frac{F_{i}^L (G(\frac{2\pi}{N})\delta_{i}^L + g(\frac{2\pi}{N})(1 -\delta_{i}^L))}{ \mathit{l}_{L}(\|u_{i}^L-x_0\|)} + \!\!\!\!\!\!\!\!\!\!\!\!\!\!\!\!\! \sum\limits_{u_{j}^N \in \Phi_{u,N}^{'} \cap S(x_0, \frac{2(m-1)\pi}{M_{ra}},\frac{2m\pi}{M_{ra}})} \!\!\!\!\!\!\!\!\!\!\!\!\!\!\!\!\!\!\! \frac{F_{j}^N (G(\frac{2\pi}{N})\delta_{j}^N + g(\frac{2\pi}{N})(1 -\delta_{j}^N)) }{ \mathit{l}_{N}(\|u_{j}^N-x_0\|)} + \frac{\sigma^2}{P_u }},
\end{align}
where $F_0$, $F_i^L$ and $F_j^N$ represent the Rayleigh fading channels from the users to the tagged BS. In addition, $\delta_{i}^L$ ($\delta_{j}^N$) is equal to 1 if the main lobe of $u_{i}^L$ ($u_{j}^N$) covers the tagged BS and $0$ otherwise. For the first three initial access protocols in Table~\ref{IAOptions}, the transmit beam directions for the users in $\Phi_u^{'}$ are decided from the cell search phase, which are assumed to be independent and uniformly distributed with $\mathbb{E}[\delta_{i}^L] = \mathbb{E}[\delta_{j}^N] = \frac{1}{N}$ for $\forall i, j$. For the omni RX protocol, $\delta_{i}^L$ ($\delta_{j}^N$) is $1$ if the beam direction of $u_{i}^L$($u_{j}^N$) is the same as the typical user, which has beamwidth $\frac{2\pi}{N}$.
Since all the fading variables are i.i.d. exponentially distributed, the PGFL of PPP can be applied to~(\ref{RA_SINR_COP_Lemma_Proof_Eq}) similar to Lemma~\ref{CSSucProbPerSectorLemma}, which completes the proof. 

\section{Proof of Lemma~\ref{CCDF_Data_SINR_Lemma}}\label{CCDF_Data_SINR_Proof_Apdx}
Since $\mathbb{P}(\text{SINR}_{DL} \geq T | e_0\delta_0 = 1) = \frac{\mathbb{P}(\text{SINR}_{DL} \geq T \cap e_0\delta_0 = 1)} {\eta_{IA}}$, we will focus on the derivation of $\mathbb{P}(\text{SINR}_{DL} \geq T \cap e_0\delta_0 = 1)$. 
Without loss of generality, we assume the $n$-th receive beam of the typical user and the $m$-th transmit beam of the tagged BS are aligned during data transmission. 
We denote by $S_i \triangleq S(o,\frac{2\pi(i-1)}{K_{cs}}, \frac{2\pi i }{K_{cs}})$ the $i$-th BS sector for $1 \leq i \leq K_{cs}$. Note that during data transmission, the typical user is able to receive from the BSs inside $S(o,\frac{2\pi(n-1)}{N},\frac{2\pi n}{N})$ under its main lobe. Therefore, there are $q = \frac{K_{cs}}{N}$ BS sectors within $S(o,\frac{2\pi(n-1)}{N},\frac{2\pi n}{N})$, which are denoted by $\tilde{S}_1, \tilde{S}_2, ..., \tilde{S}_q$, with $\tilde{S}_i = S(o,\frac{2\pi(n-1)}{N} + \frac{2\pi (i-1) }{K_{cs}}, \frac{2\pi(n-1)}{N} + \frac{2\pi i}{K_{cs}})$ for $ 1 \leq i \leq q$. 

Among all the BS sectors, we denote by $S_{i_1}, S_{i_2}, ..., S_{i_k}$ the sectors that are detected during cell search, where $1 \leq k \leq K_{cs}$. In addition, we assume $S_{i_1}, S_{i_2}, ..., S_{i_s}$ are among $\tilde{S}_1, \tilde{S}_2, ..., \tilde{S}_q$, where $\max(1,k-K_{cs}+q) \leq s \leq \min(q,k)$. 
Given $k$ and $s$, we can obtain that there are: (1) ${N\choose 1}$ choices for receive beam direction of typical user; (2) ${q \choose 1}$ choices for the BS sector containing the tagged BS among $\tilde{S}_1$ to $\tilde{S}_q$; (3) ${q-1 \choose s-1}$ number of combinations for the other $s-1$ detected BS sectors among $\tilde{S}_1$ to $\tilde{S}_q$; and (4) ${K_{cs} - q \choose k-s}$ number of combinations for the detected BS sectors that are not among $\tilde{S}_1$ to $\tilde{S}_q$. 
Thus we have the following relation: 
\begin{align}\label{Data_SINR_CCDF_Pf_Eq1}
&\mathbb{P}(\text{SINR}_{DL}  \geq \Gamma \cap e_0\delta_0 = 1) \nonumber\\
= & \int_{0}^{\infty} \sum\limits_{k=1}^{K_{cs}} \sum\limits_{s = \max(1,k-K_{cs}+q)}^{\min(q,k)} {N\choose 1} {K_{cs} - q \choose k-s} {q\choose 1}{q-1 \choose s-1}  \mathbb{P}(\text{SINR}_{DL} \geq \Gamma \cap A) P_{M_{cs},N_{cs}}^{k-s}(z,\Gamma_{cs}) \nonumber\\
&\times (1-P_{M_{cs},N_{cs}}(\Gamma_{cs}))^{K_{cs} - q - k +s} f_{Z_1}(z) {\rm d}z,
\end{align}
where $z$ in the integration represents the path loss from the typical user to the tagged BS. In addition, $A$ denotes the event that among $\tilde{S}_1 $ to $\tilde{S}_q$, $S_{i_1}$ contains the tagged BS which the typical user is successfully connected to, $S_{i_2}$ to  $S_{i_s}$ are detected during cell search, while the rest are not detected. From the definition of $A$, it is easy to obtain that:
\begin{align}\label{Data_SINR_CCDF_Pf_Eq2}
\mathbb{P}(A) = \tilde{P}_{M_{cs},N_{cs}}(z,\Gamma_{cs}) \times P_{co} \times P_{ra} (z,\Gamma_{ra}) \times P_{M_{cs},N_{cs}}^{s-1}(z,\Gamma_{cs}) \times (1-P_{M_{cs},N_{cs}}(\Gamma_{cs}))^{q-s}.
\end{align}
Since random access is an uplink procedure which does not dependent on the BS process given $z$, $\text{SINR}_{DL}$ can be expressed as follows given event $A$ happens:
\begin{align}\label{Data_SINR_CCDF_Pf_Eq3}
\text{SINR}_{DL} = \frac{P_b M G(\frac{2\pi}{N}) F_0 / z }{I_1+ I_2 + I_3 + I_4 + \sigma^2},
\end{align}
where: 
\allowdisplaybreaks
\begin{align}\label{Data_SINR_CCDF_Pf_Eq4}
\allowdisplaybreaks
\!\!\!\! I_1 &= \!\!\!\!\!\!\!\!\!\!\!\!\sum\limits_{x_i^L \in \Phi_L \cap (\cup_{j = 1}^{s} S_{i_j}) \cap B^c(o,\mathit{l}^{-1}_L(z))} \!\!\!\!\!\!\!\!\!\!\!\! P_b M G(\frac{2\pi}{N}) F_i^L \delta_{i}^{L} /\mathit{l}_L(\|x_i^L\|) + \!\!\!\!\!\! \!\!\!\!\!\!\sum\limits_{x_i^N \in \Phi_N \cap(\cup_{j = 1}^{s} S_{i_j}) \cap B^c(o,\mathit{l}^{-1}_N(z))}  \!\!\!\!\!\!\!\!\!\!\!\!\!\!\!\! P_b M G(\frac{2\pi}{N}) F_i^N \delta_{i}^{N} /\mathit{l}_N(\|x_i^N\|) , \nonumber\\
\!\!\!\! I_2 &= \!\!\!\!\!\!\!\!\!\!\!\!\!\!\!\sum\limits_{x_i^L \in \Phi_L \cap (S(o,\frac{2\pi(i-1)}{K_{cs}}, \frac{2\pi i }{K_{cs}}) \setminus (\cup_{j = 1}^{s} S_{i_j})) } \!\!\!\!\!\!\!\!\!\!\!\!\!\!\!\!\!\!\!\!\!\!\!\!\!\!\!\! P_b M G(\frac{2\pi}{N}) F_i^L \delta_{i}^{L} /\mathit{l}_L(\|x_i^L\|) + \!\!\!\!\!\!\!\!\!\!\!\!\!\!\!\! \sum\limits_{x_i^N \in \Phi_N \cap (S(o,\frac{2\pi(i-1)}{K_{cs}}, \frac{2\pi i }{K_{cs}}) \setminus (\cup_{j = 1}^{s} S_{i_j}))}  \!\!\!\!\!\!\!\!\!\!\!\!\!\!\!\! P_b M G(\frac{2\pi}{N}) F_i^N \delta_{i}^{N} /\mathit{l}_N(\|x_i^N\|),\nonumber \\
\!\!\!\! I_3 &= \!\!\!\!\!\!\!\!\!\!\!\!\sum\limits_{x_i^L \in \Phi_L \cap (\cup_{j = s+1}^{k} S_{i_j}) \cap B^c(o,\mathit{l}^{-1}_L(z))} \!\!\!\!\!\!\!\!\!\!\!\! P_b M g(\frac{2\pi}{N}) F_i^L \delta_{i}^{L} /\mathit{l}_L(\|x_i^L\|) + \!\!\!\!\!\! \!\!\!\!\!\!\sum\limits_{x_i^N \in \Phi_N \cap(\cup_{j = s+1}^{k} S_{i_j}) \cap B^c(o,\mathit{l}^{-1}_N(z))}  \!\!\!\!\!\!\!\!\!\!\!\!\!\!\!\! P_b M g(\frac{2\pi}{N}) F_i^N \delta_{i}^{N} /\mathit{l}_N(\|x_i^N\|) , \nonumber\\
\!\!\!\! I_4 &= \!\!\!\!\!\!\!\!\!\!\!\!\!\!\!\sum\limits_{x_i^L \in \Phi_L  \setminus (S(o,\frac{2\pi(i-1)}{K_{cs}}, \frac{2\pi i }{K_{cs}}) \cup (\cup_{j = s+1}^{k} S_{i_j})) } \!\!\!\!\!\!\!\!\!\!\!\!\!\!\!\!\!\!\!\!\!\!\!\!\!\!\!\! P_b M g(\frac{2\pi}{N}) F_i^L \delta_{i}^{L} /\mathit{l}_L(\|x_i^L\|) + \!\!\!\!\!\!\!\!\!\!\!\!\!\!\!\! \sum\limits_{x_i^N \in \Phi_N  \setminus (S(o,\frac{2\pi(i-1)}{K_{cs}}, \frac{2\pi i }{K_{cs}}) \cup (\cup_{j = s+1}^{k} S_{i_j}))}  \!\!\!\!\!\!\!\!\!\!\!\!\!\!\!\!\!\!\!\! P_b M g(\frac{2\pi}{N}) F_i^N \delta_{i}^{N} /\mathit{l}_N(\|x_i^N\|).
\end{align}
$I_1$ and $I_2$ ($I_3$ and $I_4$) represent the interference from BSs that user receives under its main lobe (side lobe), which come from the BS sectors that are detected and not detected during cell search respectively. In~(\ref{Data_SINR_CCDF_Pf_Eq4}), $F_0$, $F_i^L$, and $F_i^N$ represent the exponential fading variables. 
The indicators $\delta_{i}^{L}$ and $\delta_{i}^{N}$ in~(\ref{Data_SINR_CCDF_Pf_Eq4}) represent whether the transmit beam direction of the interfering BS covers the typical user or not, which happens with probability $\frac{1}{M}$. 
Therefore, based on~(\ref{Data_SINR_CCDF_Pf_Eq3}) and~(\ref{Data_SINR_CCDF_Pf_Eq4}), as well as the PGFL of PPPs, we can derive the following result:
\begin{align}\label{Data_SINR_CCDF_Pf_Eq5}
\allowdisplaybreaks
\mathbb{P}(\text{SINR}_{DL} > \Gamma | A) = &\exp\biggl(-\frac{\Gamma z\sigma^2}{P_bM G(\frac{2\pi}{N})}\biggl) \biggl(V(z,\Gamma,\frac{\lambda }{ M K_{cs}}) \biggl)^s \biggl( U(z,\Gamma,\frac{\lambda}{M K_{cs}}) \biggl)^{q-s} \nonumber \\ &\times \biggl(V(z,\frac{g(2\pi/N)}{G(2\pi/N)}\Gamma,\frac{\lambda }{ M K_{cs}}) \biggl)^{k-s} \biggl( U(z,\frac{g(2\pi/N)}{G(2\pi/N)}\Gamma,\frac{\lambda}{M K_{cs}}) \biggl)^{K_{cs}-q-k+s}.
\end{align}

For $\forall a,b,c,d \in \mathbb{R}$, we have:
\allowdisplaybreaks
\allowdisplaybreaks
\begin{align}\label{Data_SINR_CCDF_Pf_Eq6}
\allowdisplaybreaks
\!\!\!\!\!(a+b)^{q-1} (c+d)^{K_{cs}-q} = &\sum\limits_{m=0}^{K_{cs}-q} \sum\limits_{l=0}^{q-1} {K_{cs} - q \choose m} {q-1 \choose l} a^{l} b^{q-1-l} c^{m} d^{K_{cs}-q-m} \nonumber\\
\overset{(a)}{=} & \sum\limits_{k=1}^{K_{cs}} \sum\limits_{s = \max(1,k-K_{cs}+q)}^{\min(q,k)} {K_{cs} - q \choose k-s} {q-1 \choose s-1} a^{s-1} b^{q-s} c^{k-s} d^{K_{cs}-k-q+s},
\end{align}
where (a) is obtained by letting $k = m + l +1$ and $s = l+1$. Finally, the proof is concluded by substituting (\ref{Data_SINR_CCDF_Pf_Eq2}), (\ref{Data_SINR_CCDF_Pf_Eq5}) and (\ref{Data_SINR_CCDF_Pf_Eq6}) into (\ref{Data_SINR_CCDF_Pf_Eq1}).

\bibliographystyle{ieeetr}
\bibliography{reference}

\end{document}